\documentclass{article}

\usepackage{amsmath}
\usepackage{amsthm}
\usepackage{hyperref}
\usepackage{xr}
\usepackage[page]{appendix}
\usepackage{graphics}
\usepackage{graphicx}
\usepackage{amssymb}
\usepackage{stmaryrd}
\usepackage{txfonts}
\usepackage{subfloat}
\usepackage{caption}
\usepackage{subcaption}
\usepackage{url}
\usepackage{multirow}

\newtheorem{lemma}{Lemma}[section]
\newtheorem{theorem}[lemma]{Theorem}
\newtheorem{proposition}[lemma]{Proposition}

\newtheorem*{theorem*}{Theorem}

\theoremstyle{definition}

\DeclareMathOperator{\Var}{Var}
\DeclareMathOperator{\diag}{diag}

\newcommand{\Stirling}[2]{\genfrac{\{}{\}}{0pt}{}{#1}{#2}_\textrm{sf}}
\newcommand{\stirling}[2]{\genfrac{\{}{\}}{0pt}{}{#1}{#2}}

\begin{document}


  \title{Poisson PCA: Poisson Measurement Error corrected PCA, with Application
  to Microbiome Data} 
  \author{Toby Kenney\thanks{The authors gratefully acknowledge funding from NSERC},
    Tianshu Huang  \and    Hong Gu\thanks{The authors gratefully acknowledge funding from NSERC}\\
Department of Mathematics and Statistics, Dalhousie University}

  \maketitle
\begin{abstract}
In this paper, we study the problem of computing a Principal Component
Analysis of data affected by Poisson noise.  We assume samples are
drawn from independent Poisson distributions. We want to estimate
principle components of a fixed transformation of the latent Poisson
means. Our motivating example is microbiome data, though the methods
apply to many other situations. We develop a semiparametric approach
to correct the bias of variance estimators, both for untransformed and
transformed (with particular attention to log-transformation) Poisson
means. Furthermore, we incorporate methods for correcting different
exposure or sequencing depth in the data. In addition to identifying
the principal components, we also address the non-trivial problem of
computing the principal scores in this semiparametric framework. Most previous
approaches tend to take a more parametric line. For example the
Poisson-log-normal (PLN) model, approach. We compare our method with
the PLN approach and find that our method is better at identifying the
main principal components of the latent log-transformed Poisson means,
and as a further major advantage, takes far less time to compute.
Comparing methods on real data, we see that our method also appears to
be more robust to outliers than the parametric method.
\end{abstract}

\noindent%
\vfill

\section{Introduction}

Principal component analysis (PCA) is a widely used dimension
reduction and data exploration technique. It is particularly suited
for real valued multivariate data analysis with the underlying
assumption that the measurement error in the data is additive. This
can be seen from the fact that one way to define the PCA is by
minimizing the reconstruction error of the data by the first $k$
principal components with the objective function defined as squared
error loss. For other types of data, this loss function can be
natually extended to be the negative loglikelihood function. This can
be most natually developed to the whole exponential family
distribution, for example, Collins {\em et al.} (2001) and Chiquet
{\em et al.} (2018) maximized the exponential family loglikelihood of the data
by assuming a lower dimensional representation of the natural
parameters through the canonical link functions. Such an approach can
be deemed as totally parametric and thus is sensitive to the model
assumptions, outliers and the choice of the dimension for the latent
parameter space. To some extent, this is more a parametric modelling
method instead of data exploration.

Our focus in this paper is to develop a new PCA method for count data
by correcting Poisson measurement errors, especially for those data
observed from the high throughput sequencing technologies.
Measurement error is an important issue in statistical analysis, that
occurs when we want to analyse some latent variable, but our attempts
to measure the variable are subject to random error. Our framework is
conceptually different from that of Collins {\em et al.} (2001) and
Chiquet {\em et al.} (2018). A closely related approach is the work of
Liu {\em et al.} (2018) who develop the same estimator in the
untransformed case without sequencing depth noise, as we do, but then
proceed to develop shrinkage estimators to improve performance, rather
than deal with sequencing depth noise or transformations. We assume
that the latent Poisson means do not lie exactly on a low-dimensional
space, but instead follow some unknown distribution, and we wish to
estimate the principal components of this distribution or of a non-linear
transformation of this distribution.  When no transformation is
needed, a PCA on an unbiased estimate of the covariance matrix of the
latent Poisson means is straightforward. When a non-linear
transformation is applied on the latent Poisson means, our idea is to
derive a non-linear transformation of the data from which an unbiased
estimator of the covariance matrix of the transformed latent Poisson
means can be derived and thus the PCA can be applied on this unbiased
estimate of the covariance. Our motivation for extending our work to
transformed latent Poisson means is that previous research suggests
that a logarithmic scale is an appropriate way to study microbial
abundance. This also makes sense from the point of view of
differential equation models of species abundance, which tend to imply
exponential-type growth.  We have avoided any parametric constraints
on this latent distribution. Although our research is driven by the
application to microbiome data which we will describe in more detail
in the following subsection, the method is geneally applicable to high
dimensional count data which are popularly observed in many different
areas including social science, ecology, microbiology and bioscience
in general.

The structure of the paper is as follows. In
Section~\ref{MicrobiomeFeatures}, we introduce the features of
microbiome data which motivated this research. In
Section~\ref{Method}, we develop the bias corrected estimators for the
covariance matrices of functions of the latent Poisson means. Then in
Section~\ref{NonlinearSequencingDepth} we develop methods for handling
the sequencing depth issue, particularly in the $\log$-transformed
case. In Section~\ref{Projection}, we provide a method for projecting
the data onto the principal component space. We then demonstrate the
performence of our methods on simulated data in
Section~\ref{Simulations} and on a real microbiome data set (Caporaso
{\em et al}, 2010) in Section~\ref{RealData}. Finally in
Section~\ref{Conclusions}, we summarise our work and suggest promising
directions for future research.

The methods in this paper are all implemented in the \texttt{R}
package \texttt{PoissonPCA}, which is available from the first
author's website \url{www.mathstat.dal.ca/~tkenney/PoissonPCA/}, and
will be submitted to \texttt{CRAN}.

\section{Features of microbiome data}\label{MicrobiomeFeatures}
The microbiome is the collection of all microorganisms in a particular
environment. This collection mostly consists of bacteria, archaea,
eukaryotes and viruses. Recent improvements in technology have allowed
the bacteria and archaea components to be measured en-masse. That is,
identifying DNA from all bacteria and archaea present in a given
sample can be collected and counted, with the objective of giving a
complete picture of the whole microbial community. This picture can
then be compared to metadata, such as whether the microbial community
was from a healthy individual or from an individual suffering from a
particular illness. The idea is that the interaction of the thousands
of microbes present in the whole community may be more informative
than simple analyses that try to attribute complex conditions to the
actions of a single species. 
There have been a large number of papers that
have identified associations between the microbiome and various environmental 
conditions such as nutrient abundance (Arrigo 2005), and human health
(Fujimura {\em et al.} 2010, Sekirov {\em et al.} 2010) including 
obesity (Turnbaugh {\em et al.} 2009), Crohn’s disease (Quince {\em et al.} 2013), 
diabetes (Kostic {\em et al.} 2015). 

The microbiome data present a number of interesting data analysis
challenges that make the application of PCA difficult.  The feature of microbiome data that is the focus of this paper is
measurement error. While each environment has a community with a
certain total abundance of each 
Operational Taxonomic
Unit (OTU), we only observe a small sample
from that community, which for this paper, we will model as following
a Poisson distribution with mean given by the underlying total
abundance. In practice the sequencing procedure is more complicated
(there are several different sequencing procedures that can be used,
each of which involves their own noise and bias --- Gorzelak {\em et al.} (2015) give more details about how
the sequencing procedure influences the data) but the Poisson noise
error seems a reasonable starting point. The aim of this paper is to
adapt PCA to account for this Poisson measurement error.

A further feature of microbiome data is sequencing depth. Essentially,
the issue is that the Poisson means of our observed sample are subject
to large multiplicative noise. For this reason, microbiome data is
usually treated as compositional data, with only the proportions of
each OTU considered, for example Li (2015). However, this loses a lot
of information about the Poisson distribution. For the Poisson
distribution, the relative noise in the observed proportions will
decrease as sequencing depth increases, so samples at higher
sequencing depth should be given more weight. Another commonly used
approach to the sequencing depth problem is rarefaction, where large
parts of the data are thrown away, so that all samples have the same
total OTU count. This practice is mainly due to the use of arbitrary
indices which are sensitive to sample size. McMurdie and Holmes (2014)
argue strongly against the use of rarefaction.

The other features of microbiome data include the high dimensionality
and sparsity. There is substantial literature on ways to deal with
these issues by imposing sparsity constraints on the PCA. For example
Zou {\em et al.} (2006), and Jolliffe {\em et al.} (2003). Salmon {\em
  et al.} (2014) combine the Poisson PCA of Collins {\em et al.}
(2001) with an $l^1$ penalty to create a sparse Poisson PCA, and apply
this method to Photon-limited image denoising.  However, the
sparseness issue is not the focus for the current paper, and we will
avoid it by grouping our OTUs at a higher taxonomic level,
i.e. instead of considering variables that roughly correspond to the
abundance of a particular species of microbe, we will look at
higher-level variables that correspond to the total abundance of a
whole genus or higher taxonomic level. This will reduce the dimension
enough to avoid the curse of dimensionality for the datasets we
consider. In principle the $l^1$ penalty could be applied to our
method to provide a sparse Poisson-error corrected PCA, but that is a
topic for future work.

There has been other work incorporating the Poisson error structure of
microbiome data into the analysis, but it has been focused on
parametric modelling, rather than Exploratory data analysis. For
example Holmes {\em et al.} (2012) model microbiome data as a
Dirichlet multinomial distribution; Knights {\em et al.} (2011)
use a Bayesian mixture model; Shafiei {\em et al.} (2014) and Shafiei {\em et al.}  (2015) develop hierarchical Bayesian multinomial mixture models; and Cai {\em et al.} (2017) apply
non-negative matrix factorisation with Poisson error.

\section{Method}\label{Method}

 We express the situation more formally as
follows: We have $n$ samples of a $p$-dimensional random vector $\mathbf{X}$
which follows a Poisson distribution with conditional mean given by a
latent random vector $\mathbf{\Lambda}$ and a ``sequencing
depth'' $S$ (this comes from our specific application to microbiome
data). That is, we have $\mathbf{X}_i|\mathbf{\Lambda}_i, S \sim Po(\mathbf{\Lambda}_iS)$,  $\mathbf{X}_i$'s $(i=1,\cdots, p)$ are
conditionally independent given the random vector $\mathbf{\Lambda}$
and the sequencing depth $S$. 
$\mathbf{X}$ can be thought of as $\mathbf{\Lambda}$ with some Poisson measurement
error. Our samples form an $n\times p$ matrix $X$ of observed data,
where the $i$th row is an independent sample of the random vector
$\mathbf{X}$. We want to estimate the distribution
of $\mathbf{\Lambda}$ and its relationship to other variables. For this
purpose, we develop a method of estimating a PCA on this latent
$\mathbf{\Lambda}$, or on some transformation of $\mathbf{\Lambda}$, $f(\mathbf{\Lambda})$. Our approach to
this problem is moment-based --- we aim to correct the bias in the
variance estimates, and use the unbiased variance estimates to perform
PCA.

We start with the linear case, where we want to estimate the PCA of
$\mathbf{\Lambda}$. There are two cases to consider: the case without
sequencing depth noise (where $S_i=1$ for all $i=1, \cdots, n$), and the case where
sequencing depth correction is needed. We then develop the methods for estimating the PCA of a nonlinear transformation $f(\mathbf{\Lambda})$.  

\subsection{Poisson Noise Correction in PCA without Sequencing Depth Corrections}\label{LinearCase}

We first consider the simplest case, where $S_i=1$ for all $i=1, \cdots, n$. That
is, any differences in total abundance between samples is a result of
differences in underlying latent abundances. The model here is fairly
straightforward. Latent Poisson means are a random vector $\mathbf{\Lambda}$,
and observed counts are a random vector $\mathbf{X}$ with
$\mathbf{X}_i|\mathbf{\Lambda}_i\sim Po(\mathbf{\Lambda}_{i})$ conditionally independent given $\mathbf{\Lambda}$. We want to
estimate the principal components of the random vector $\mathbf{\Lambda}$ from
a matrix $X$ whose rows are independent samples of the random vector
$\mathbf{X}$. The principal components of $\mathbf{\Lambda}$ are the eigenvectors of the variance
of the vector $\mathbf{\Lambda}$, so we only need to provide an estimate for
the variance of the random vector $\mathbf{\Lambda}$. We have
\begin{align*}
  \Var(\mathbf{X})&={\mathbb E}\left(\Var(\mathbf{X}|\mathbf{\Lambda})\right)+\Var\left({\mathbb
    E}(\mathbf{X}|\mathbf{\Lambda})\right)\\
&={\mathbb E}\left(\diag(\mathbf{\Lambda})\right)+\Var\left(\mathbf{\Lambda}\right)\\
\Var\left(\mathbf{\Lambda}\right)&= \Var(\mathbf{X})-{\mathbb
  E}\left(\diag(\mathbf{\Lambda})\right)\\
&= \Var(\mathbf{X})-\diag\left({\mathbb E}(\mathbf{X})\right)\\
\end{align*}
where $\diag(\mathbf{\Lambda})$ is a diagonal matrix with entries
$\mathbf{\Lambda}$.
Now we have the following straightforward estimators for $\Var(\mathbf{X})$ and
$\diag\left({\mathbb E}(\mathbf{X})\right)$:
\begin{align*}
  \widehat{{\mathbb E}(\mathbf{X})}&=\frac{1}{n}X^T1\\
  \widehat{\Var(\mathbf{X})}&=\frac{X^TX}{n-1}-\frac{X^T11^TX}{n(n-1)}\\
\end{align*}
plugging these in gives the estimator 
$$\widehat{\Var\left(\mathbf{\Lambda}\right)}=\frac{1}{n-1}X^TX-\frac{1}{n(n-1)}X^T11^TX-\diag\left(\frac{1}{n}X^T1\right)$$

\subsection{Poisson Noise Correction in PCA with Sequencing Depth Corrections}

Suppose we now add sequencing depth noise to the problem, so that the
Poisson mean is $S\mathbf{\Lambda} $, where $S$ is a random scalar and
$\mathbf{\Lambda}$ is a random vector. To avoid the identifiability issue we assume $\mathbf{\Lambda}^T 1=1$, and that $\mathbf{\Lambda}$ and
$S$ are independent. We can apply the same procedure
\begin{align*}
  \Var(\mathbf{X}|S)&={\mathbb E}\left(\Var(\mathbf{X}|\mathbf{\Lambda},S)\right)+\Var\left({\mathbb
    E}(\mathbf{X}|\mathbf{\Lambda},S)\right)\\
&={\mathbb E}\left(\diag(\mathbf{\Lambda} S)|S\right)+\Var\left(\mathbf{\Lambda} S|S\right)\\
&=S{\mathbb E}\left(\diag(\mathbf{\Lambda})|S\right)+S^2\Var\left(\mathbf{\Lambda}|S \right)
\end{align*}
Thus 
\begin{align*}
\Var\left(\mathbf{\Lambda}\right)={\mathbb E}\Var\left(\mathbf{\Lambda}|S \right)&= {\mathbb E}\Var\left(\frac{\mathbf{X}}{S}\middle |S\right)-{\mathbb E}{\mathbb
  E}\left(\diag\left(\frac{\mathbf{\Lambda}}{S}\right)\middle|S\right)\\
&= {\mathbb E}\Var\left(\frac{\mathbf{X}}{S}\middle | S\right)-{\mathbb E}\diag\left({\mathbb E}\left(\frac{\mathbf{X}}{S^2}\right)\middle|S\right)\\
&={\mathbb E}\Var\left(\frac{\mathbf{X}}{S}\middle | S\right)-\diag\left({\mathbb E}\left(\frac{\mathbf{X}}{S^2}\right)\right)\\
\end{align*}
Now since $\mathbf{X}_i|\mathbf{\Lambda}_i, S \sim Po(\mathbf{\Lambda}_iS)$, $(i=1,\cdots, p)$, we have 
\begin{align*}
\Var\left(\frac{\mathbf{X}}{S}\right)&={\mathbb
  E}\left(\Var\left(\frac{\mathbf{X}}{S}\middle | S\right)\right)+\Var\left({\mathbb E}\left(\frac{\mathbf{X}}{S}\middle |
S\right)\right)\\
&={\mathbb  E}\left(\Var\left(\frac{\mathbf{X}}{S}\middle |
S\right)\right)+\Var\left(\frac{{\mathbb E}(\mathbf{\Lambda}) S}{S}\middle |S\right)\\
&={\mathbb  E}\left(\Var\left(\frac{\mathbf{X}}{S}\middle |
S\right)\right)
\end{align*}
Thus 
\begin{align*}
\Var\left(\mathbf{\Lambda}\right)&= \Var\left(\frac{\mathbf{X}}{S}\right)-\diag\left({\mathbb E}\left(\frac{\mathbf{X}}{S^2}\right)\right)\\
\end{align*}
We can now empirically estimate this quantity from a sample. Suppose
we have a matrix $X$, where each row is a sample from the underlying
distribution of $\mathbf{X}$. If the sequencing depth $S$ is known, then we can
compute $\frac{X}{S}$ for each sample, and we have empirical
estimators
\begin{align*}
  \widehat{\Var\left(\frac{\mathbf{X}}{S}\right)}&=\frac{1}{n-1}\left(X^TD^{-2}X-\frac{1}{n}X^TD^{-1}11^TD^{-1}X\right)\\
\widehat{{\mathbb
    E}\left(\frac{\mathbf{X}}{S^{2}}\right)}&=  \frac{1}{n}X^TD^{-2}1\\
\end{align*}
where $D$ is a diagonal matrix with entries the sequencing depths
of each observation. Plugging this into our formula for
$\Var(\mathbf{\Lambda})$ gives us the estimator

\begin{align*}
\widehat{\Var\left(\mathbf{\Lambda}\right)}&= \frac{1}{n-1}\left(X^TD^{-2}X-\frac{1}{n}X^TD^{-1}11^TD^{-1}X\right)-\frac{1}{n}\diag\left(X^TD^{-2}1\right)\\
\end{align*}

The correction here is similar to the correction in (Li, Palta and
Shao, 2004), where they added a correction term
$\frac{{\beta_1}^2}{2{\sigma_\epsilon}^2}\sum_{i=1}^N
\frac{W_i}{{L_i}^2}$ to the likelihood term in a regression model,
where $W_i$'s are the Poisson samples, and $L_i$'s are the sequencing
depths. Their case is a regression model, and in their case the
sequencing depth is more directly observed (being number of hours of sleep).

\subsection{Poisson PCA on Transformed Poisson Means without Sequencing Depth Corrections}\label{TransformedMeans}

Suppose now that we are interested in the covariance matrix of some
transformation $f(\mathbf{\Lambda})$, where $f$ is some univariate
function applied elementwise to $\mathbf{\Lambda}$. For example, we
might set $f(\mathbf{\Lambda})=\log(\mathbf{\Lambda})$. We can use the
same technique as in Section~\ref{LinearCase}, but with some
univariate transformation $g(\mathbf{X})$ of the original data. Using
the law of total variance, we have:
$$\Var(g(\mathbf{X}))={\mathbb
  E}\left(\Var(g(\mathbf{X})|\Lambda)\right)+\Var\left({\mathbb
  E}(g(\mathbf{X})|\Lambda)\right)$$
Now if we can find $g(\mathbf{X})$ so that ${\mathbb
  E}(g(\mathbf{X})|\mathbf{\Lambda})=f(\mathbf{\Lambda})$ then we get
$$\Var(f(\mathbf{\Lambda}))=\Var(g(\mathbf{X}))-{\mathbb
  E}\left(\Var(g(\mathbf{X})|\mathbf{\Lambda})\right)$$

The first step, therefore is
finding $g(\mathbf{X})$ such that ${\mathbb E}(g(\mathbf{X})|\mathbf{\Lambda})=f(\mathbf{\Lambda})$. Without loss of generality, we take $\mathbf{\Lambda}$ as univariate. We can directly calculate ${\mathbb
  E}(g(\mathbf{X})|\mathbf{\Lambda})$ as follows:
\begin{align*}
  {\mathbb E}(g(\mathbf{X})|\mathbf{\Lambda}=\lambda)&=e^{-\lambda}\sum_{n=0}^\infty \frac{\lambda^n g(n)}{n!}= f(\lambda)\\
 e^{\lambda}f(\lambda)&=\sum_{n=0}^\infty  \frac{g(n)}{n!} {\lambda}^n\\
  \end{align*}


That is, we can take $\frac{g(n)}{n!}$ to be the coefficients of the
Taylor series of $e^{\lambda}f(\lambda)$. As an example, suppose we
want $g(\mathbf{X})$ to give an unbiased estimator for
$f(\mathbf{\Lambda})=\mathbf{\Lambda}^k$ for some fixed $k$. We have 
\begin{align*}
 e^{\lambda}\lambda^k&=\lambda^k\sum_{n=0}^\infty  \frac{\lambda^n}{n!} 
 =\sum_{m=k}^\infty  \frac{\lambda^{m}}{m!}\times\frac{m!}{(m-k)!} \\
 \end{align*}
We therefore get our estimator $g(\mathbf{X})=\mathbf{X}(\mathbf{X}-1)\cdots(\mathbf{X}-k+1)$, which is a
well-known unbiassed estimator for $\lambda^k$. 

Next we need to calculate the conditional variance
$\Var(g(\mathbf{X})|\mathbf{\Lambda})$. Since the elements of $\mathbf{X}$ are conditionally independent given 
$\mathbf{\Lambda}$, this matrix is a diagonal matrix and it will suffice to take $\mathbf{\Lambda}$ as univariate for the estimation of each element of this matrix.  We can estimate this from first principles:
\begin{align*}
  \Var(g(\mathbf{X})|\mathbf{\Lambda}=\lambda)&={\mathbb E}(g(\mathbf{X})^2|\mathbf{\Lambda}=\lambda)-{\mathbb E}(g(\mathbf{X})|\mathbf{\Lambda}=\lambda)^2\\
  &=e^{-\lambda}\sum_{n=0}^\infty
  \frac{\lambda^n}{n!}g(n)^2-e^{-2\lambda}\left(\sum_{n=0}^\infty \frac{\lambda^n}{n!}g(n)\right)^2\\
  &=\sum_{n=0}^\infty g(n)^2  e^{-\lambda}\frac{\lambda^n}{n!}-e^{-2\lambda}\left(\sum_{n=0}^\infty\sum_{m=0}^\infty \frac{\lambda^n}{n!}\frac{\lambda^m}{m!}g(m)g(n)\right)\\
  &=\sum_{n=0}^\infty g(n)^2
  e^{-\lambda}\frac{\lambda^n}{n!}-\sum_{n=0}^\infty h(n) e^{-2\lambda}\frac{\lambda^n}{n!}\\
\end{align*}
where $h(n)=\sum_{m=0}^n \binom{n}{m}g(m)g(n-m)$ is the coefficient of
the Taylor series $e^{2\lambda}f(\lambda)^2=\sum_{n=0}^\infty h(n)\frac{\lambda^n}{n!}$.
Since $\mathbf{\Lambda}$ is latent, we need to replace the expressions 
$e^{-\lambda}  \frac{\lambda^n}{n!}$ and $e^{-2\lambda}
\frac{\lambda^n}{n!}$ by estimators which are functions of $\mathbf{X}$. Since
we will be taking the expectation of our estimators of the conditional
variance, it makes sense to focus on the bias of these estimators. For
$e^{-\lambda}  \frac{\lambda^n}{n!}$, there is an unbiassed estimator
$$s_n(\mathbf{X})=\left\{\begin{array}{ll}1&\textrm{if }\mathbf{X}=n\\0&\textrm{otherwise}\end{array}\right.$$
which can reasonably be used. This also has a computational advantage
because it results in only needing to compute a single term in the
sum. For the term $e^{-2\lambda}  \frac{\lambda^n}{n!}$, the unbiassed
estimator is $$t_n(\mathbf{X})=(-1)^{\mathbf{X}-n}{\mathbf{X}\choose n}$$
This gives us the overall estimator for the variance of $f(\mathbf{\Lambda_l})$, where $\mathbf{\Lambda_l}$ is an element of the vector $\mathbf \Lambda$.
\begin{align}
\nonumber
\widehat{\Var\left(f(\mathbf{\Lambda_l})\right)}&=\widehat{\Var\left(g(\mathbf{X_l})\right)}-
\widehat{{\mathbb E_{\mathbf{\Lambda_l}}}\Var\left(g(\mathbf{X_l}|\mathbf{\Lambda_l})\right)}\\
\nonumber&=\frac{1}{n-1}\sum_{i=1}^ng(X_{il})^2-\frac{1}{n(n-1)}\left(\sum_{i=1}^n
g(X_{il})\right)^2-\frac{1}{n}\sum_{i=1}^n\left(\sum_{k=0}^\infty g(k)^2s_k(X_{il})-\sum_{k=0}^\infty
h(k)t_k(X_{il})\right) \\
  &=\frac{1}{n}\sum_{i=1}^n\sum_{k=0}^{X_{il}}
(-1)^{X_{il}-k}{X_{il}\choose k}h(k)-\frac{1}{n(n-1)}\left(\sum_{\stackrel{i\ne j}{i,j=1}}^n g(X_{il})g(X_{jl})\right)\label{TransformedVarianceEstimator}
\end{align}
%
The estimates of the off-diagonal elements of the covariance matrix can be computed as the corresponding 
elements of the sample covariance matrix of $g(X)$, where $g$ is applied elementwise
to the observed data matrix $X$.

In Appendix~\ref{ConsistencyProof}, we show the consistency of this estimator:

\begin{theorem}\label{Consistency}
If $e^{2\lambda}f(\lambda)^2=\sum_{n=0}^\infty
\frac{h(n)\lambda^n}{n!}$ is a globally convergent Taylor series,
$\mathbf{\Lambda}$ is a non-negative random variable with finite raw moments $\mu_n$
and $\sum_{n=0}^\infty\frac{\mu_n|h_n|}{n!}$ converges, then the
estimator \eqref{TransformedVarianceEstimator} is consistent.
\end{theorem}

In practice, this estimator can have high variance, leading to poor
finite-sample performance for certain transformations. Sometimes this
issue can be alleviated by using suitable approximations, and
simplifying the algebra in
equation~\eqref{TransformedVarianceEstimator}. We show how this can be
done for the commonly used log transformation in
Appendix~\ref{Computingh}. We also suggest a possible approach to the
more general case, that could be given further study in Appendix~\ref{generallowvarianceCVestimator}.

\subsection{The Special Case of Log Transformation}

In the special case when $f(\lambda)=\log(\lambda)$, the approach from
Section~\ref{TransformedMeans} does not work, because $f$ does
not have a globally convergent Taylor series. We can compute a Taylor
series for $\log(\lambda)=\log\left(a\left(1+\frac{\lambda-a}{a}\right)\right)$
for some value of $a$. That is, we get
$$\log(\lambda)=\log(a)+\sum_{n=1}^\infty
\frac{(-1)^n}{n}\left(\frac{\lambda}{a}-1\right)^n$$ This Taylor
series is convergent for $0<\lambda<2a$, however, when we expand the terms
and try to reorganise it as a power series in $\lambda$, the coefficients do
not converge. Therefore, we can only use a truncated power
series. That is, we truncate the above series at some chosen value
$N$, to get a polynomial. This produces a reasonable approximation to
$\log(\lambda)$ on the interval $(0,2a)$. For larger values of
$\lambda$, we have that ${\mathbb E}(\log(\mathbf{X})|\mathbf{\Lambda}=\lambda)\approx\log(\lambda)$. Therefore,  we set $g(X)=\log(X)$ instead for larger values of $X$. The
result of this is that ${\mathbb E}(g(\mathbf{X})|\mathbf{\Lambda})\approx\log(\mathbf{\Lambda})$
for $\mathbf{\Lambda}$ not too close to zero. 
In summary our function $g(n)$ is given by
$$g(n)=\left\{\begin{array}{ll}\log(a)+\sum_{i=1}^N\frac{1}{i}&\textrm{if }n=0\\
\log(a)+\sum_{i=1}^N\frac{1}{i}+\sum_{i=1}^{n\land N}\frac{(-1)^in!}{a^i(n-i)!}\sum_{m=i}^N\frac{\binom{m}{i}}{m}&\textrm{if }1\leqslant n<l_0\\
\log(n)&\textrm{if }n\geqslant l_0
\end{array}\right.$$

There is some difficulty choosing
the values of $a$ and $N$. From experiments, we see that $a=3, N=4$ and $g(X)=\log(X)$ whenever
$X\geqslant 7$, gives a fairly good estimator. We discuss the
corresponding estimator for conditional variance in
Appendix~\ref{Computingh}.

\section{Sequencing Depth Correction}\label{NonlinearSequencingDepth}

The methods described above for variance estimators of transformed
latent Poisson means cover the situation without sequencing depth
correction. In this section, we give two approaches to correcting the
resulting covariance matrix for sequencing depth.

\subsection{Method 1: Compositional Covariance Matrix for Sequencing Depth Correction}\label{CompositionalCovariance}
The first method is to get a
compositional covariance matrix. That is, suppose the estimated
covariance matrix is $\Sigma$. We want to find the matrix
$\tilde{\Sigma}$ with the following properties:

\begin{enumerate}

 \item For any $u,v$ with $u^T1=v^T1=0$, we have
   $u^T\tilde{\Sigma}v=u^T\Sigma v$

 \item $\tilde{\Sigma}$ is symmetric.

 \item $\tilde{\Sigma}1=0$
   
\end{enumerate}

Conditions~2 and 3 are necessary for $\tilde\Sigma$ to be the
covariance matrix of compositional data. Condition~1 says that the
covariance matrices $\tilde\Sigma$ and $\Sigma$ give the same
covariance for the projection of $f(\mathbf{\Lambda})$ onto the space
orthogonal to the vector $1$. Now to solve this, we have the
following.

\begin{proposition}\label{CompositionalSeqDepthCorrection}
  The matrix $\tilde\Sigma_{p \times p}$ that satisfies Conditions~1--3 is given by
  $$\tilde\Sigma=\Sigma-a1^T-1a^T$$
  where $a=(pI+11^T)^{-1}\Sigma 1$.
\end{proposition}

\begin{proof}
We first note that for any $v$ with $v^T1=0$, we have
$\tilde{\Sigma}v-\Sigma v=\alpha(v) 1$ for some scalar $\alpha(v)$. This
is because $u^T\left(\tilde{\Sigma}v-\Sigma v\right)=0$ for any $u$
with $u^T1=0$. By linearity, $\alpha(v)$ is a linear function of $v$,
so we have $\alpha(v)=a^Tv$ for some fixed vector $a$. Now we have
$\left(\tilde\Sigma-\Sigma-1a^T\right)v=0$ for all $v$ with
$v^T1=0$. If we let $\Sigma_r=\tilde\Sigma-\Sigma-1a^T$, then this
gives $\Sigma_r x=\Sigma_r \frac{1}{p}11^Tx$ because
$\frac{1}{p}11^Tx$ is the orthogonal projection of $x$ onto the
1-dimensional space spanned by 1. Therefore we have $\Sigma_r
=\Sigma_r \frac{11^T}{p}=b1^T$ where $b=\frac{\Sigma_r 1}{p}$. Thus we
have $\tilde\Sigma-\Sigma-1a^T=b1^T$ so
$\tilde\Sigma=\Sigma+1a^T+b1^T$ Since $\tilde\Sigma$ and $\Sigma$ are
both symmetric, we get $a=b$, which gives the form in the
proposition. All that now remains is to find the value of $a$. This is
forced by $\tilde\Sigma 1=0$ which gives
\begin{align*}
\left(\Sigma-a1^T-1a^T\right)1&=0\\
\Sigma 1-a1^T1-1a^T 1&=0\\
\Sigma 1-pa-11^T a&=0\\
(pI+11^T)a&=\Sigma 1\\
a&=(pI+11^T)^{-1}\Sigma 1
\end{align*}
\end{proof}

We have that $(pI+11^T)(cI+d11^T)=cpI+(c+pd+pd)11^T$, so solving for
$cp=1$ and $c+2pd=0$ we get that
$(pI+11^T)^{-1}=\frac{I}{p}-\frac{11^T}{2p^2}$, so
$$\widetilde{\Sigma}=\Sigma-a1^T-1a^T=\Sigma-\frac{1}{p}\Sigma
11^T+\frac{11^T}{2p^2}\Sigma 11^T-11^T\frac{\Sigma}{p}+11^T\Sigma\frac{11^T}{2p^2}=\left(I-\frac{11^T}{p}\right)\Sigma\left(I-\frac{11^T}{p}\right)$$

\subsection{Method 2: Minimum Variance by Correcting for the Maximum Sequencing Depth Noise}\label{MinvarCorrection}

The second method is specific to the $\log(\mathbf{\Lambda})$ case. In this case,
the sequencing depth becomes additive noise on $\log(\mathbf{\Lambda})$,
proportional to the vector 1. This means that $\Sigma=\Sigma^*+c11^T$.
The value of $c$ is not identifiable, but it is limited by the
constraint that $\Sigma^*$ should be non-negative definite. One
natural choice is therefore to select the largest value of $c$ such
that $\Sigma^*$ remains non-negative definite. In the case where $1$
is an eigenvector of $\Sigma$, this will give exactly the same result
as the compositional method. In other cases, the results may differ
slightly. 

To solve for the largest value of $c$ such that $\Sigma-c11^T$ is non-negative definite, we consider the determinant
$\left|\Sigma-c11^T\right|$ as a continuous function of $c$. While $\Sigma-c11^T$ is positive
definite, the determinant is positive. We therefore set $c$ to be the
smallest solution to $\left|\Sigma-c11^T\right|=0$. We solve this with
a change of basis $\Sigma=ASA^T$ where $A$ is an orthogonal matrix
with $Ae_1=\frac{1}{\sqrt{p}}1$, where $e_1=(1, 0, \cdots, 0)^T$. Now we have that
\begin{align*}
  \left|\Sigma-c11^T\right|&=\left|A\left(S-cpe_1{e_1}^T\right)A^T\right|\\
  &=\left|S-cpe_1{e_1}^T\right|\\
  &=\left|S\right|-cp\left|S_{\hat{1}\hat{1}}\right|\\
\end{align*}
where $S_{\hat{1}\hat{1}}$ is the matrix $S$ with the first row and
column removed and the first entry replaced by $1$. For example, if
$$S=\left(\begin{array}{lll}4 & 5 & 6\\5 & 7 & 9\\6 & 9 &
  12\end{array}\right)$$
  then $$S_{\hat{1}\hat{1}}=\left(\begin{array}{lll}1 & 0 & 0\\0 & 7 & 9\\0 & 9 &
  12\end{array}\right)$$
  More formally $S_{\hat{1}\hat{1}}=S-e_1s^T-s{e_1}^T+(S_{11}+1)e_1{e_1}^T$ where
  $s=Se_1$ is the first column of $S$. This gives
\begin{align*}
  S_{\hat{1}\hat{1}}&=
  S-e_1{e_1}^TS-Se_1{e_1}^T+({e_1}^TSe_1+1)e_1{e_1}^T\\
  &=S-e_1{e_1}^TS-Se_1{e_1}^T+e_1{e_1}^TSe_1{e_1}^T+e_1{e_1}^T\\
\end{align*}
This gives us that 
\begin{align*}
  AS_{\hat{1}\hat{1}}A^T&=A\left(S-e_1{e_1}^TS-Se_1{e_1}^T+e_1{e_1}^TSe_1{e_1}^T+e_1{e_1}^T\right)A^T\\
  &=\Sigma-\frac{1{1}^T}{p}\Sigma-\Sigma \frac{11^T}{p}+\frac{11^T}{p}\Sigma \frac{11^T}{p}+\frac{11^T}{p}\\
  &=\left(I-\frac{1{1}^T}{p}\right)\Sigma\left(I-\frac{1{1}^T}{p}\right)+\frac{11^T}{p}\\
  &=\widetilde{\Sigma}+\frac{11^T}{p}
\end{align*}
where $\widetilde{\Sigma}$ is the compositional estimate from
Section~\ref{CompositionalCovariance}. So we are trying to solve for $c$ such that
\begin{align*}
\left|\Sigma\right|-cp\left|\widetilde{\Sigma}+\frac{11^T}{p}\right|&=0
\end{align*}
Thus 
\begin{align*}
c=\frac{\left|\Sigma\right|}{p\left|\widetilde{\Sigma}+\frac{11^T}{p}\right|}
\end{align*}

In practice, this method does not perform well, because the small
eigenvalues and eigenvectors are not well estimated. Indeed some of
the estimates are often negative. We therefore improve practical
performance by introducing some thresholding.
Let $\Sigma=VDV^T$ where $$D=\left(\begin{array}{llll}\lambda_1&0 &
  \cdots & 0\\0 &\lambda_2 & \cdots &0 \\ \vdots & \vdots & \ddots
  & \vdots \\ 0 & 0 & \cdots & \lambda_p\end{array}\right)$$
with $\lambda_1>\lambda_2>\cdots>\lambda_p$.
Since $\Sigma$ should be non-negative definite, we first replace any
negative eigenvalues by 0. Next we choose the smallest $j$ such that
$\frac{\lambda_2+\cdots+\lambda_j}{\lambda_2+\cdots+\lambda_p}>0.9$. The
idea here is that we assume that $\lambda_1$ is influenced
mostly by sequencing depth and take eigenvalues accounting for 90\% of
the variance we are trying to estimate. We now replace all smaller
eigenvalues $\lambda_k$ for $k>j$ by $\lambda_{j}$. We now have a new
diagonal matrix
$$\widetilde{D}=\left(\begin{array}{llllll}\lambda_1&0 &  \cdots & 0&
  \cdots & 0\\0 &\lambda_2 & \cdots &0 &  \cdots & 0\\ \vdots & \vdots & \ddots
  & \vdots & & \vdots \\ 0 & 0 & \cdots & \lambda_j & \cdots &
  0 \\ \vdots & \vdots & 
  & \vdots & \ddots & \vdots \\ 0 & 0 & \cdots & 0 & \cdots &
\lambda_j   \end{array}\right)$$
Then we apply the above minimum variance
sequencing depth correction procedure to $V\widetilde{D}V^T$.

\section{Projection onto Principal Component Space}\label{Projection}

Having estimated the principal components of the latent Poisson means,
the remaining problem is to ``project'' the observations onto this
principal component space. The methods in this paper are based on a
semiparametric framework. We are assuming that the latent Poisson
means follow some unknown distribution, and are seeking to use
non-parametric methods to estimate the principal components of this
latent distribution. In this spirit, our projection method should be
based on the same assumptions. We will assume for each observed point,
that there is some latent random vector $\mathbf{\Lambda}$ of Poisson means,
and that we want to project these Poisson means onto the principal
component space. In the non-latent variance of principal components,
we have an observation $X$, and we seek the projection $X^*$ in
principal component space that minimises $(X-X^*)^T(X-X^*)$. In our
case, we have a latent random vector $\mathbf{\Lambda}$, and have calculated
the principal component space of $f(\mathbf{\Lambda})$. If we knew $\mathbf{\Lambda}$,
we would seek the projection $R$ in the principal component space that
minimises
$\left(f(\Lambda)-R\right)^T\Sigma^{-1}\left(f(\Lambda)-R\right)$, where
$\Sigma$ is the covariance matrix of $f(\mathbf{\Lambda})$. (Because the principal
component space is spanned by eigenvectors of $\Sigma$, it turns out
that it is equivalent to take the projection $R$ that minimises
$\left(f(\Lambda)-R\right)^T\left(f(\Lambda)-R\right)$.)

However, we do not observe $\mathbf{\Lambda}$ directly, instead, we observe $\mathbf{X}\sim
Po(\mathbf{\Lambda})$. We therefore simultaneously expect $\mathbf{\Lambda}$ to be such
that $\mathbf{X}$ has high likelihood, and such that $f(\mathbf{\Lambda})$ is close
to the principal component space. We therefore create the objective
function
$$L(\Lambda)=\frac{1}{2}\left(f(\Lambda)-R\right)^T\Sigma^{-1}\left(f(\Lambda)-R\right)+\Lambda^T1-\mathbf{X}^T\log(\Lambda)$$
where $\log(\Lambda)$ is the vector whose $i$th element is
$\log(\Lambda_i)$ --- that is it is the logarithm function applied
elementwise to the vector $\Lambda$. Note that because we have
computed the covariance matrix of $f(\Lambda)$, there is a canonical
choice of how to balance the two constraints, and we do not need a
tuning parameter. If the latent variable $f(\Lambda)$ is normal, then
$L(\Lambda)$ is the negative log-likelihood of the joint distribution of $\mathbf{X}$ and $\mathbf{\Lambda}$
(treating $R$ as parameters).

To solve this, we substitute the linear model $f(\Lambda)=\mu+Va$
where $V$ is the matrix of principal component vectors.  That is $\Sigma=VDV^T$. Suppose that our
projection is onto the first $r$ principal components and let
$$(a_{[r]})_i=\left\{\begin{array}{ll}a_i&\textrm{if }i\leqslant r\\0&\textrm{otherwise}\end{array}\right.$$
That is, $a_{[r]}$ is the vector of the first $r$ elements of
$a$. Then we have $R=\mu+Va_{[r]}$. In this basis, the objective function becomes
$$L(a)=\frac{1}{2}\left(a-a_{[r]}\right)^TD^{-1}\left(a-a_{[r]}\right)+f^{-1}\left(\mu+Va\right)^T1-\mathbf{X}^T\log\left(f^{-1}\left(\mu+Va\right)\right)$$
where $D$ is the diagonal matrix of eigenvalues of $\Sigma$, in the
order corresponding to $V$.  The estimate for
$a$ is therefore given by setting the derivative of this to zero. That
is
$$0=\nabla L(a)=D^{-1}\left(a-a_{[r]}\right)+V^T\left(\nabla
f\left(\Lambda\right)^T\right)^{-1}1-V^T\left(\nabla
f\left(\Lambda\right)^T\right)^{-1}\diag(\Lambda)^{-1} X$$
In the particular case where $f(\Lambda)=\log(\Lambda)$ (evaluated elementwise)
this equation becomes
$$D^{-1}\left(a-a_{[r]}\right)+V^T(\Lambda-X)=0$$ Differentiating
again, the Hessian matrix for the case where $f(\Lambda)=\log(\Lambda)$ is $$H=D^{-1}I_{i\geqslant
  r+1}+V^T\diag(\Lambda)V=V^T(\diag(\Lambda)+S)V$$ where $S=VD^{-1}I_{i\geqslant
  k}V^T$ Newton's method now gives
\begin{align*}
  a_{\textrm{new}}&=a-V^T(\diag(\Lambda)+S)^{-1}V\left(D^{-1}\left(a-a_{[r]}\right)+V^T\Lambda-V^TX\right)\\
  &=a-V^T(\diag(\Lambda)+S)^{-1}\left(VD^{-1}\left(a-a_{[r]}\right)+\Lambda-X\right)\\
\end{align*}

%

\section{Simulation}\label{Simulations}

\subsection{Simulation Design}

We assess the performance of the method under three different
transformations --- the default untransformed Poisson means, a
polynomial transformation $f(x)=0.004x^3-0.4x^2+16x$, and a
logarithmic transformation. In each case, we set the eigenvalues of
the covariance matrix to a fixed sequence $d$, and simulate the
eigenvectors $V$ uniformly.  We now simulate the true transformed
covariance matrix as $M=VDV^T$, where $D$ is a diagonal matrix with
diagonal entries $d$. We simulate an $n\times p$ data matrix $T$ of
transformed Poisson means, where each row of $T$ is i.i.d. with a
chosen distribution (usually multivariate normal) with covariance
matrix $M$. We compute the Poisson means from $T$ by
$\Lambda_{ij}=f^{-1}(T_{ij})$.  Finally, we simulate $X_{ij}\sim
Po(\Lambda_{ij})$. For each scenario, we simulate 100 datasets. We use
the same eigenvalues, but different eigenvectors in the 100 simulated
covariance matrices. We also simulate different means for $T$ for each
simulated dataset, using a normal distribution with mean 100 and
standard deviation 15.

For simulations with sequencing depth correction, in the compositional
case, we modify the above procedure by setting $d_p=0$, and $v_p=1$,
with other eigenvectors in $V$ simulated orthogonally to 1. In the
non-compositional case, simulate $\Lambda$ as above. In both cases, we
then simulate sequencing depth noise following a gamma distribution
with shape parameter 4 and scale parameter 100.

\subsection{Simulation Results}

\subsubsection{Assessing Simulation Results}

We want to devise a method to assess the performance of our corrected
PCA on the simulated data. In this case, we have a true covariance
matrix $\Sigma=VDV^T$, where $V$ is orthogonal and $D$ is
diagonal. The principal component matrix $V$ is defined by the
property of simultaneously maximising $\sum_{i=1}^k (V^T\Sigma
V)_{ii}$ for all $k$, subject to being orthogonal. We have obtained an
estimated covariance matrix
$\widehat{\Sigma}=\widehat{V}\widehat{D}\widehat{V}^T$ with
$\widehat{V}$ orthogonal and $\widehat{D}$ diagonal. We want to assess
the performance of $\widehat{V}$ as an estimate for $V$. Since our
objective is to maximise the variance explained by the first $k$
principal components, we will measure this by comparing cumulative
true variance explained by the first $k$ estimated principal
components, i.e. the cumulative sums
$\sum_{i=1}^k (\widehat{V}^T\Sigma\widehat{V})_{ii}$, with the
true cumulative sum  $\sum_{i=1}^k(V^T\Sigma V)_{ii}$.

\subsubsection{Simulation results for untransformed PCA without sequencing depth correction}

The simulation results for normally distributed Poisson means without
sequencing depth noise with data dimensions equal to 80 are shown in
Figure~\ref{linearsimresults}. We compare the percentages of the total
variance explained by the true eigenvectors, the eigenvectors found by
our Poisson PCA method, and the eigenvectors found by ``naive PCA''
i.e. applying PCA directly to the observed data, without any
correction for Poisson error. We see that when $n=100$, our method is
comparable with naive PCA, but as sample size increases, our method
performs much better, as expected. For large $n$, our method appears
to be converging to the correct eigenvalues as $n\rightarrow\infty$,
while uncorrected PCA is not improving with larger sample sizes,
because of the bias. Similar results when $p=30$ are shown in
Supplementary Figure~\ref{linearsimresultsp30}.

\begin{figure}[htbp]
\caption{Comparison of methods for untransformed PCA with Poisson
  noise, without sequencing depth noise}\label{linearsimresults}
Poisson PCA (blue) is compared with standard PCA (red) over a range of
sample sizes. The truth is shown in black. Results are averaged over
100 simulations. Dimension of the data is $p=80$.

\begin{subfigure}{0.48\textwidth}
  \includegraphics[width=7.5cm]{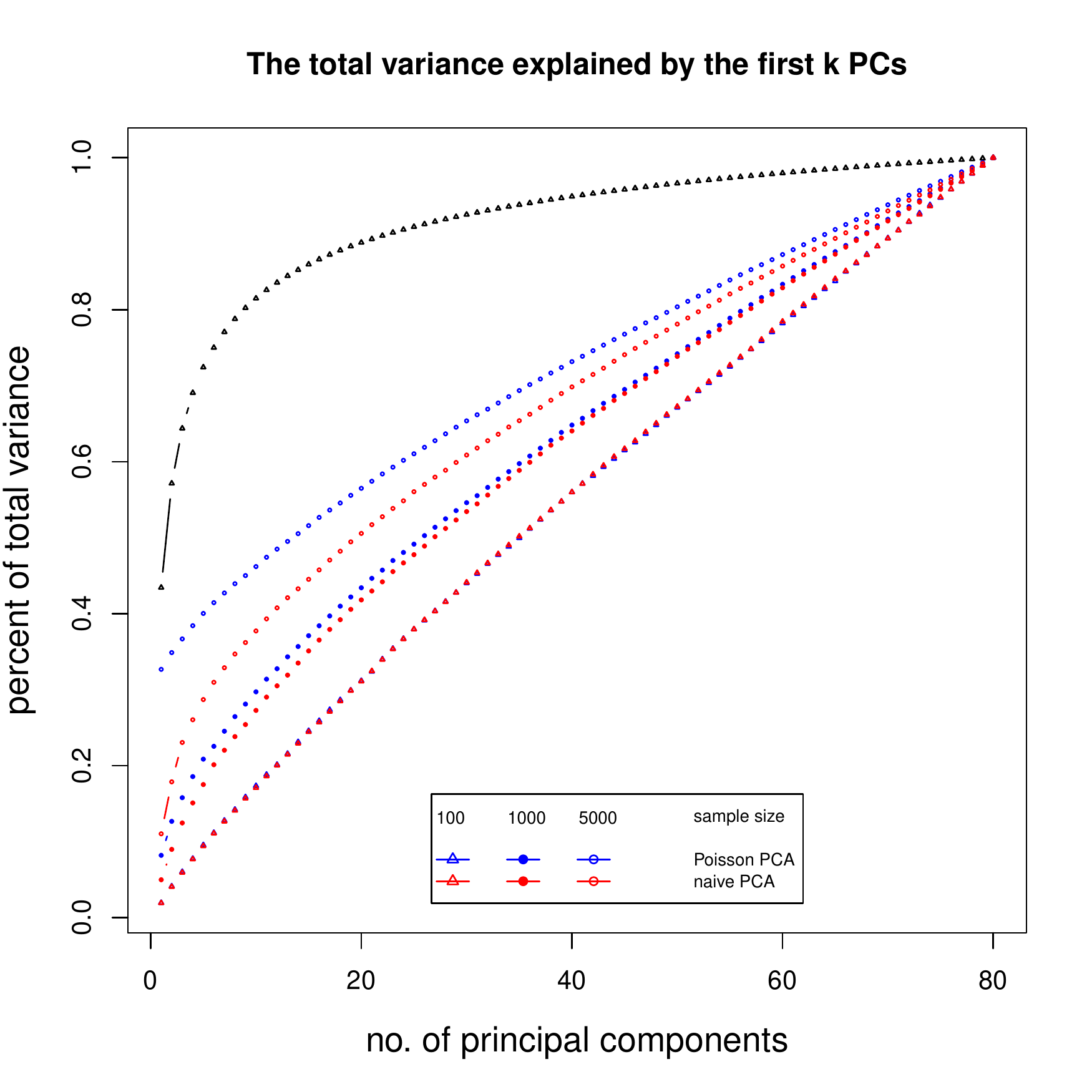}
  \caption{Sample sizes between 100 to 5000}
\end{subfigure}
\begin{subfigure}{0.48\textwidth}
  \includegraphics[width=7.5cm]{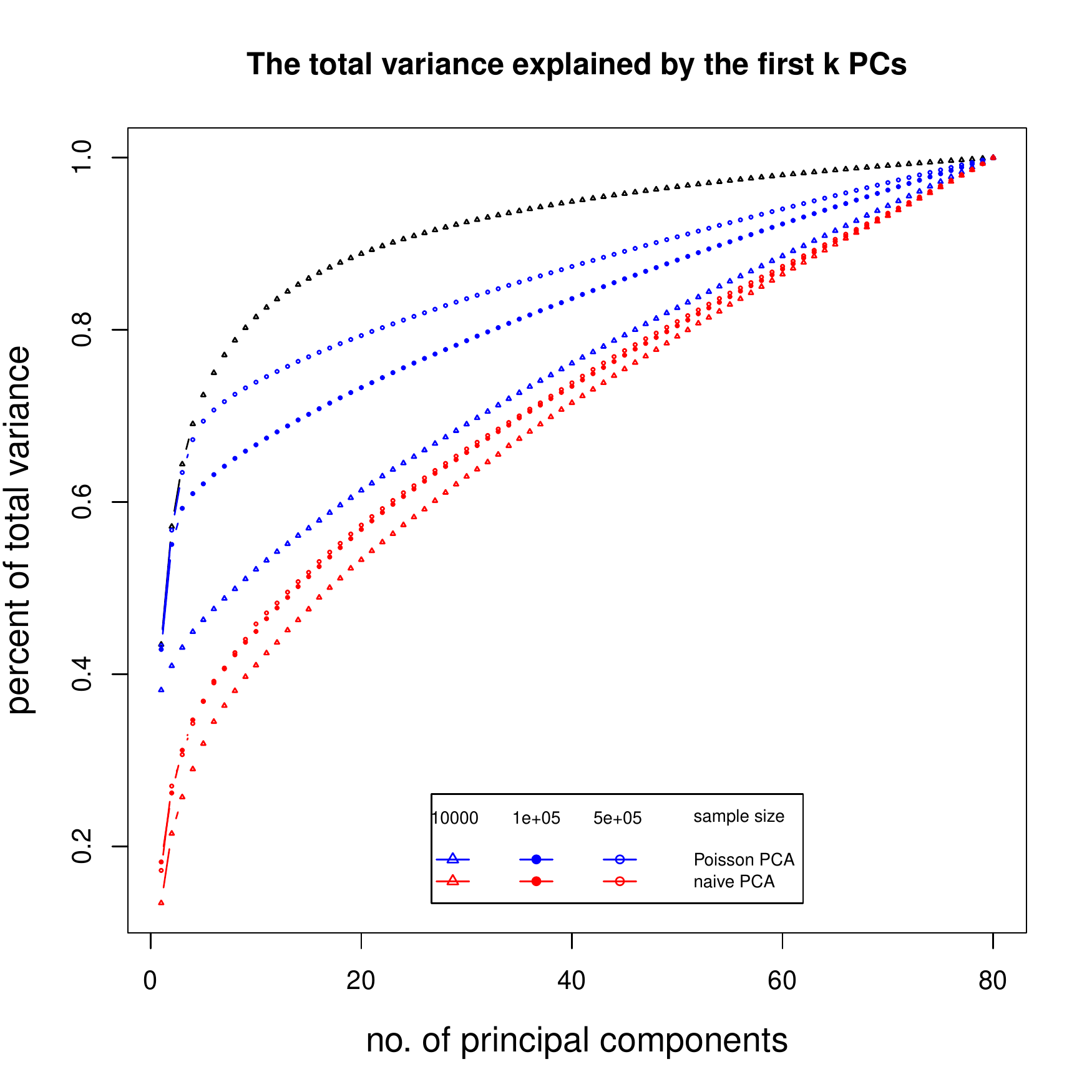}
  \caption{Sample sizes between $10^4$ to $5 \times 10^5$}
\end{subfigure}
\end{figure}

\subsubsection{Simulation results for Polynomial Transformed PCA without sequencing depth correction}

Figure~\ref{polynomialsimresultslarge} shows the simulation results
for polynomial-transformed Poisson means. Again, our method is
comparable to other methods for small sample sizes, and converges
towards the truth as sample size increases. In this case, the larger
variance of the estimators means that covergence of our method is
slower than in the untransformed case. However, Poisson PCA begins to
outperform other methods for sample size around 5,000, which is a
plausible sample size for some real data sets.  Supplemental
Figure~\ref{polynomialsimresultsp30} shows similar results for $p=30$.

\begin{figure}[htbp]

\caption{Comparison of methods for estimating polynomial transformed
  PCA with Poisson noise, without sequencing depth noise}  \label{polynomialsimresultslarge}
We compare Poisson PCA (blue) with standard PCA (red) and PCA on the
results of applying the polynomial transformation directly to the
observed data without any Poisson error correction (green).
The truth is shown in black. Results are averaged over 100
simulations. Dimension of the data is $p=80$.
  
\begin{subfigure}{0.48\textwidth}
  \includegraphics[width=7.5cm]{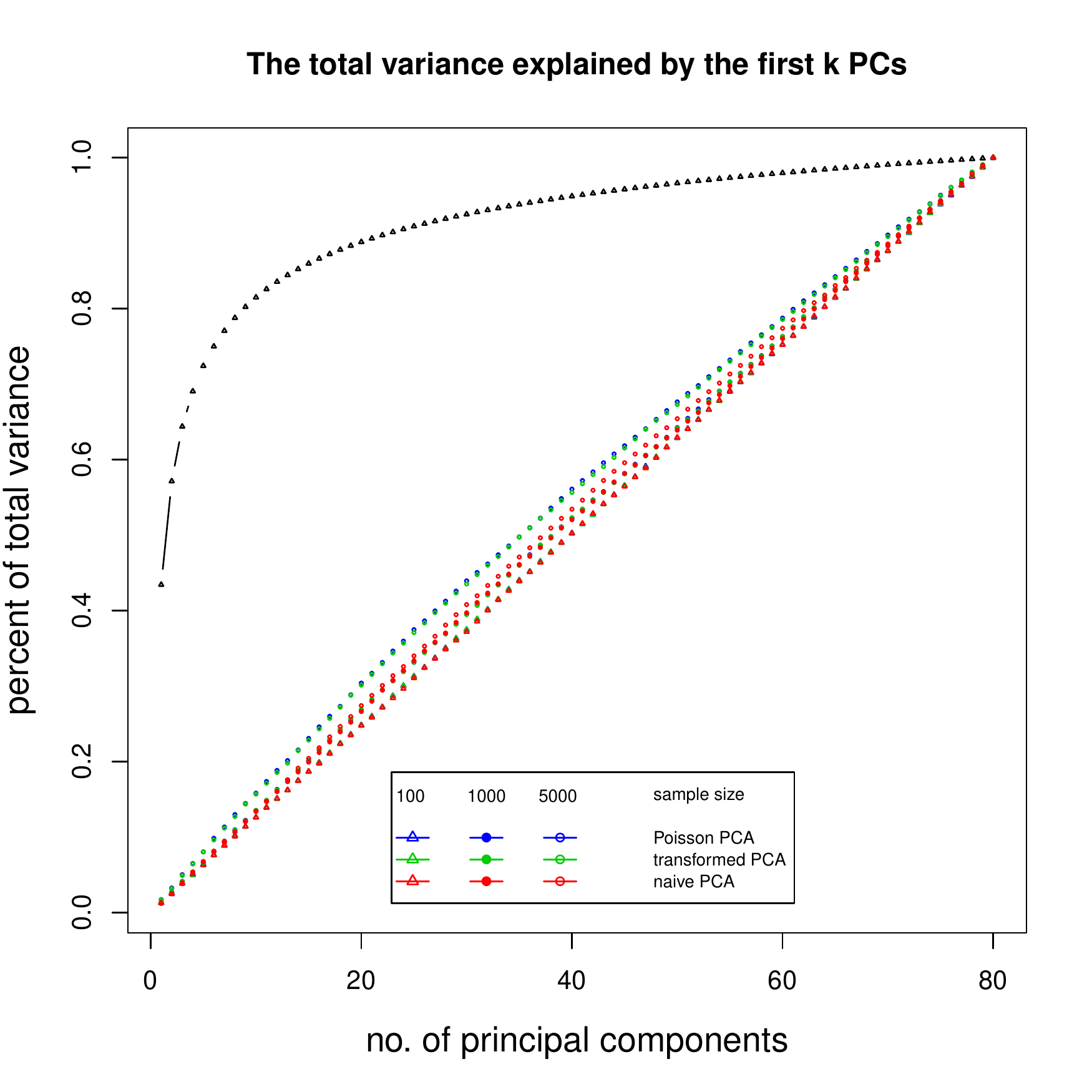}
  \caption{Sample sizes between 100 and 5000}
\end{subfigure}
\begin{subfigure}{0.48\textwidth}
  \includegraphics[width=7.5cm]{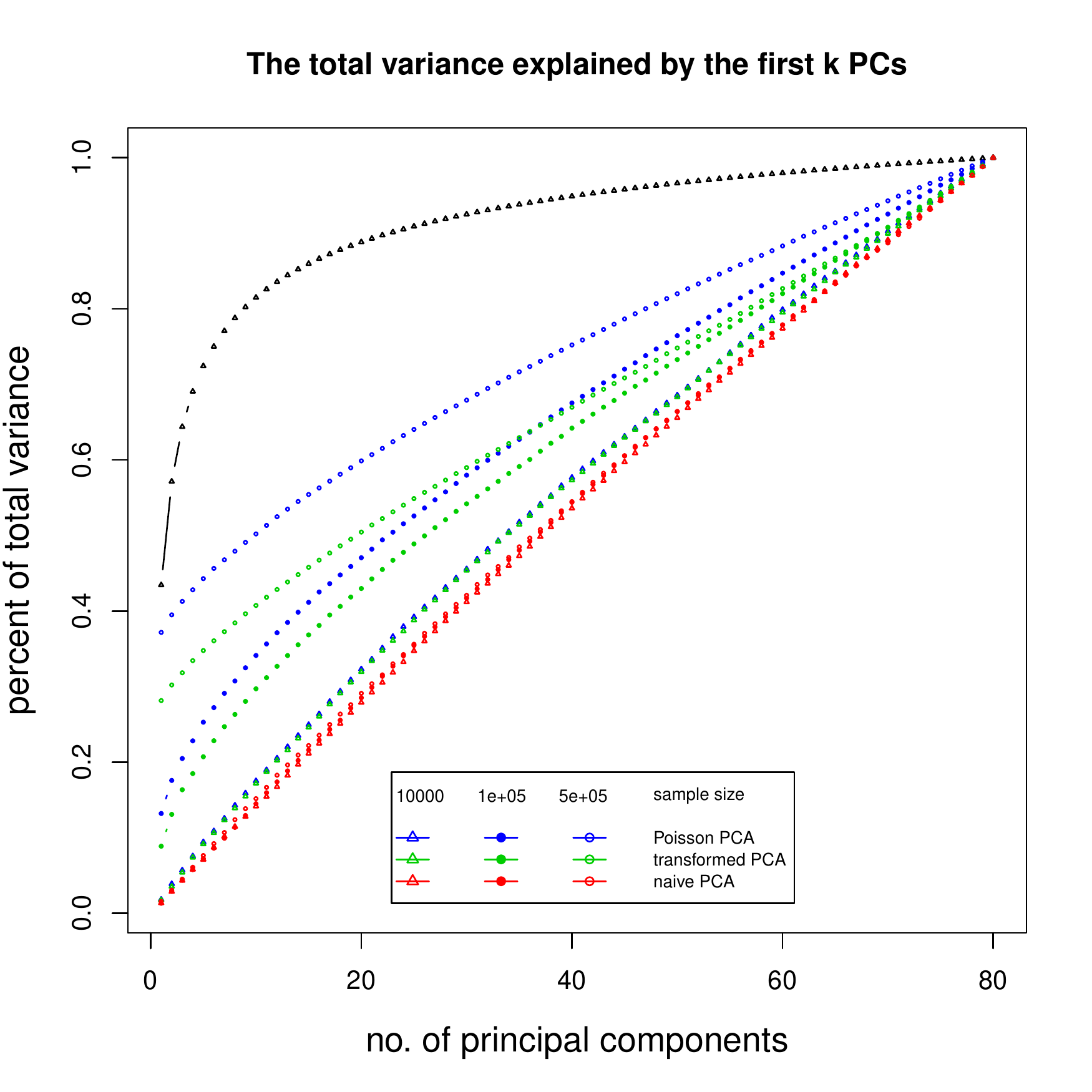}
  \caption{Sample sizes between $10^4$ and $5 \times 10^5$}
\end{subfigure}
  
\end{figure}

\subsubsection{Simulation results for Log Transformed PCA without sequencing depth correction}\label{SimResultsLogTransformed}

Figure~\ref{logsimresults} compares Poisson PCA to PLNPCA (Chiquet
{\em et. al}, 2018), which is a parametric method based on a Poisson
log-normal assumption using a variational approximation to estimate
the parameters by MLE, and to PCA directly performed on the
log-transformed data (with 0 counts arbitrarily replaced by 0.1). We
compare the methods on both log-normal and non log-normal Poisson
means. For the non log-normal case, we simulated some data in polar
coordinates with direction uniformly distributed over the sphere and
radius distributed following a centred gamma distribution with shape
$\alpha=5$ and scale $\theta=\frac{1}{\sqrt{5}}$. The data were then
stretched to have the proscribed variance. Similar results with $p=30$
are shown in Supplemental Figure~\ref{logsimresultsp30}.

When sample size is small, Poisson PCA clearly outperforms both PLNPCA
methods and the naive PCA on log-tranformed data in the first few
dimensions which capture about 70\% to 80\% of the total variance.
After this point, the total variance explained by naive PCA on
log-tranformed data starts to get higher. Between the two PLNPCA
methods, the full rank log-Normal Poisson PCA seems to outperform the
log-normal Poisson model with rank chosen by BIC. When sample size is
large ($n=1,000$) the Poisson PCA and two PLNPCA methods show
comparable results and perform better than the naive PCA on
log-tranformed data.  In practice, when we use PCA, we aim in
particular to identify the first few principal components, so better
performance at identifying the earlier prinicipal components is a very
strong asset of our method. All the methods are robust to changing the
underlying distribution of the Poisson means. Results for $n=5,000$
are shown in Supplemental Figure~\ref{logsimresultsn5000}; we do not
give results for larger $n$ because of the computation time needed by
PLNPCA.

\begin{figure}[htbp]

\caption{Comparison of methods for estimating log-transformed PCA
  without sequencing depth noise}\label{logsimresults} We compare
Poisson PCA with two Parametric Log-normal Poisson Methods (pink and
green) and naive PCA on log-transformed data (red) with 0 values
replaced by 0.1. The truth is shown in black. Results are averaged
over 100 simulations. $p=80$. The log-transformed Poisson means follow
a multivariate normal distribution in the left plot and a spherically
symmetric gamma distribution in the right.
  
\begin{subfigure}{0.48\textwidth}
  \includegraphics[width=7.5cm]{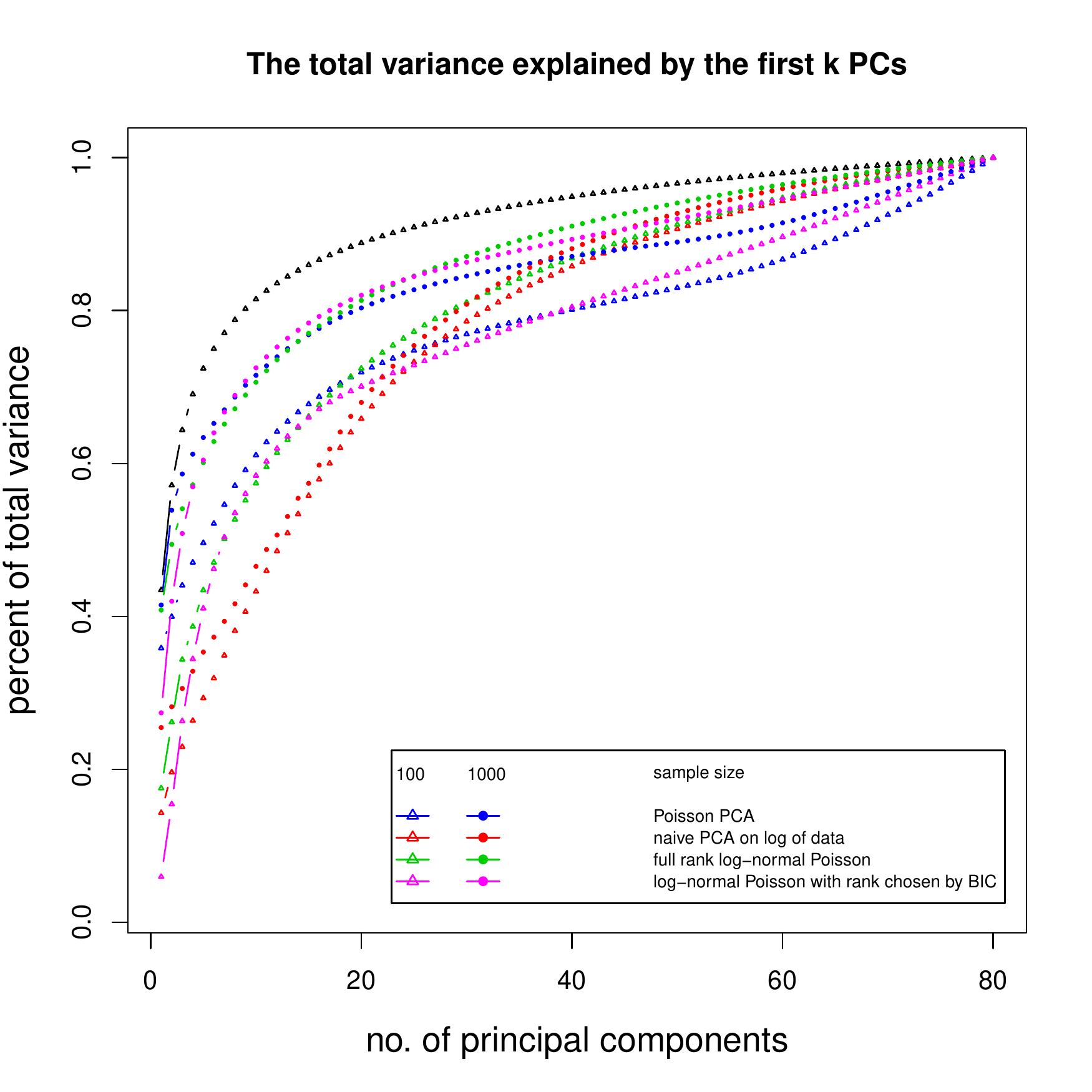}
  \caption{Normal distribution}
\end{subfigure}
\begin{subfigure}{0.48\textwidth}
  \includegraphics[width=7.5cm]{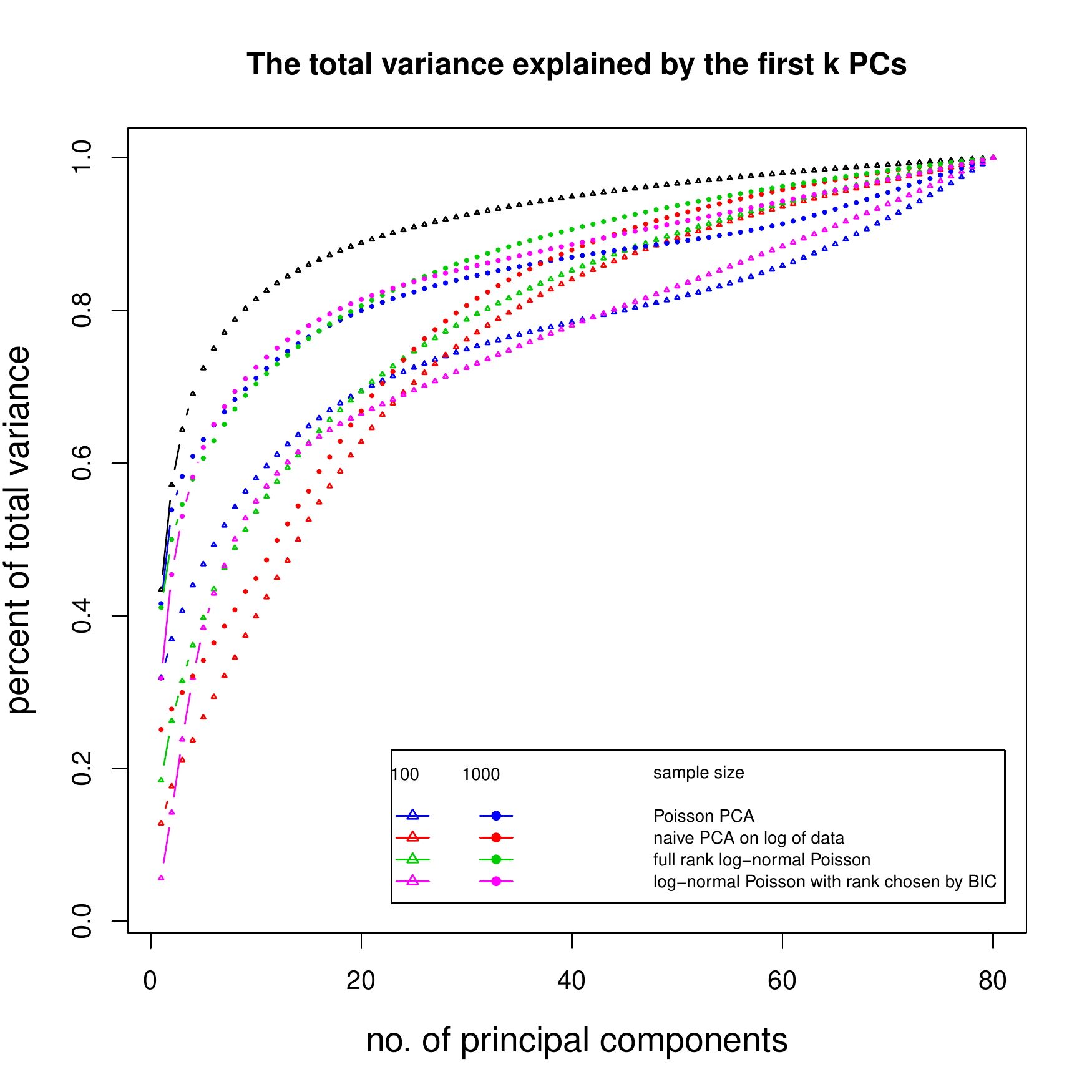}
  \caption{Gamma distribution}
\end{subfigure}
  
\end{figure}

In addition to the better identification of the early principal
components, Poisson PCA also has a major advantage in computation
time. Figure~\ref{loglogtimes} compares the logarithm of the average
processor time (amount of computation performed) over 100 simulated
datasets, for the three methods. For the limited rank PLNPCA, with
rank chosen by BIC, we compared all possible ranks. Limiting the rank
would improve the runtime of this method, to make it more comparable
to the direct PLNPCA estimate. Poisson PCA is faster than the
likelihood-based method by an order of magnitude. The runtimes are
approximately linear in $log(p)$, so we have shown regression
lines. Our Poisson PCA method is not quite linear, probably because
the singular value decomposition, performed after variance estimation,
takes nonlinear time. Because Poisson PCA is much faster than the
other methods, this nonlinearity is more apparent for our method.

\begin{figure}[htbp]

\caption{Comparison of Runtimes}\label{loglogtimes}
  
\includegraphics[width=13cm]{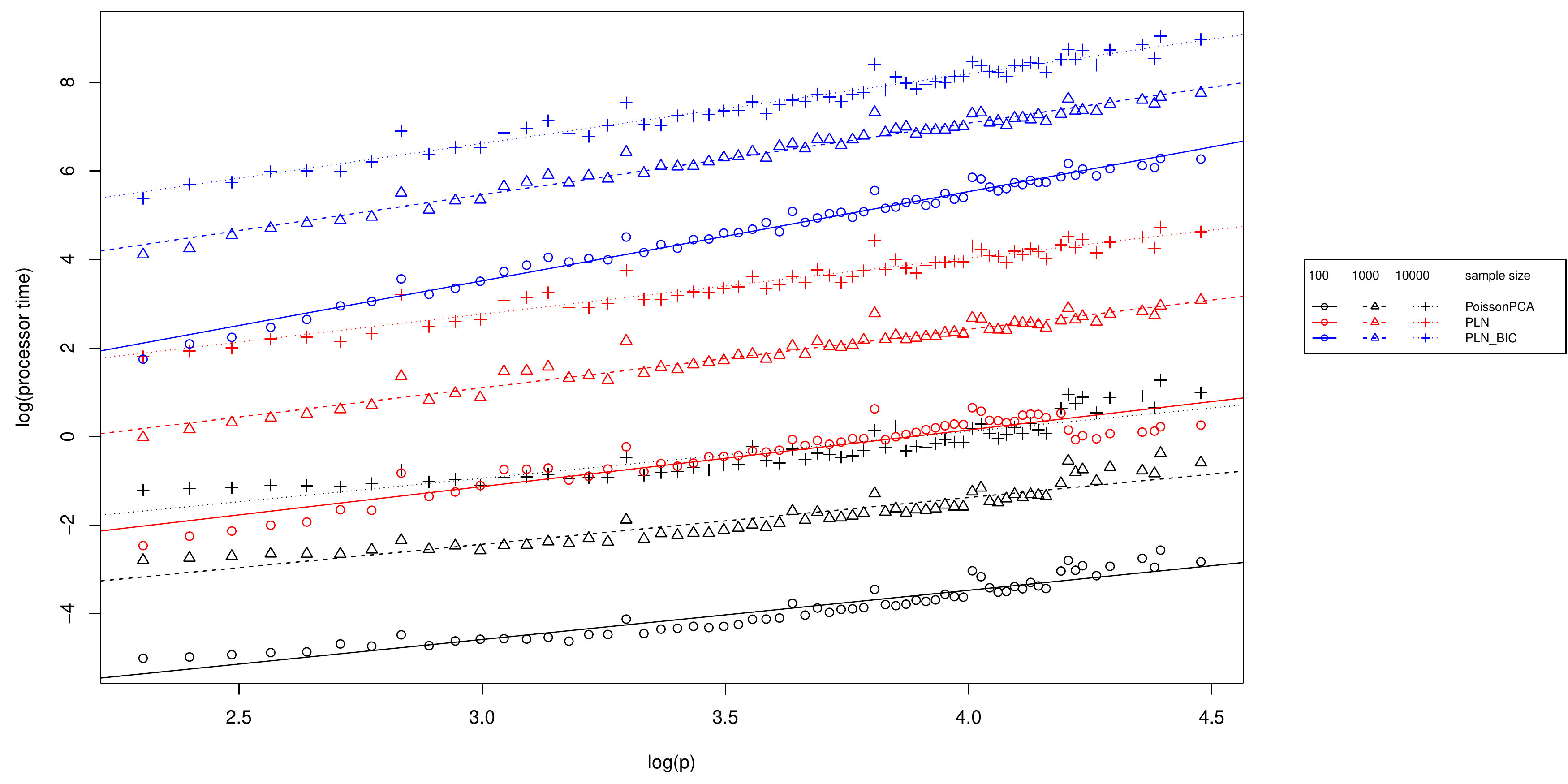}
  
\end{figure}

\subsubsection{Simulation results for untransformed PCA and log-transformed PCA with sequencing depth correction}

Figure~\ref{linearsimresultsseqdepth} compares Poisson PCA with
compositional PCA (i.e. data are rescaled so each row has total 1, and
PCA is applied to the rescaled data). The scaled Poisson means are
compositional, and follow a multivariate normal distribution; but we
have multiplied each Poisson mean by a random sequencing depth. Our
semiparametric method has done better than compositional PCA at
identifying the principal components.

\begin{figure}[htbp]

\caption{Comparison of PCA methods with sequencing depth noise and
  Poisson noise}\label{linearsimresultsseqdepth}
Poisson PCA (blue) is compared with standard PCA (red) over a range of
sample sizes. The truth is shown in black. Results are averaged over
100 simulations. Dimension of the data is $p=80$.

\begin{subfigure}{0.48\textwidth}
  \includegraphics[width=7.5cm]{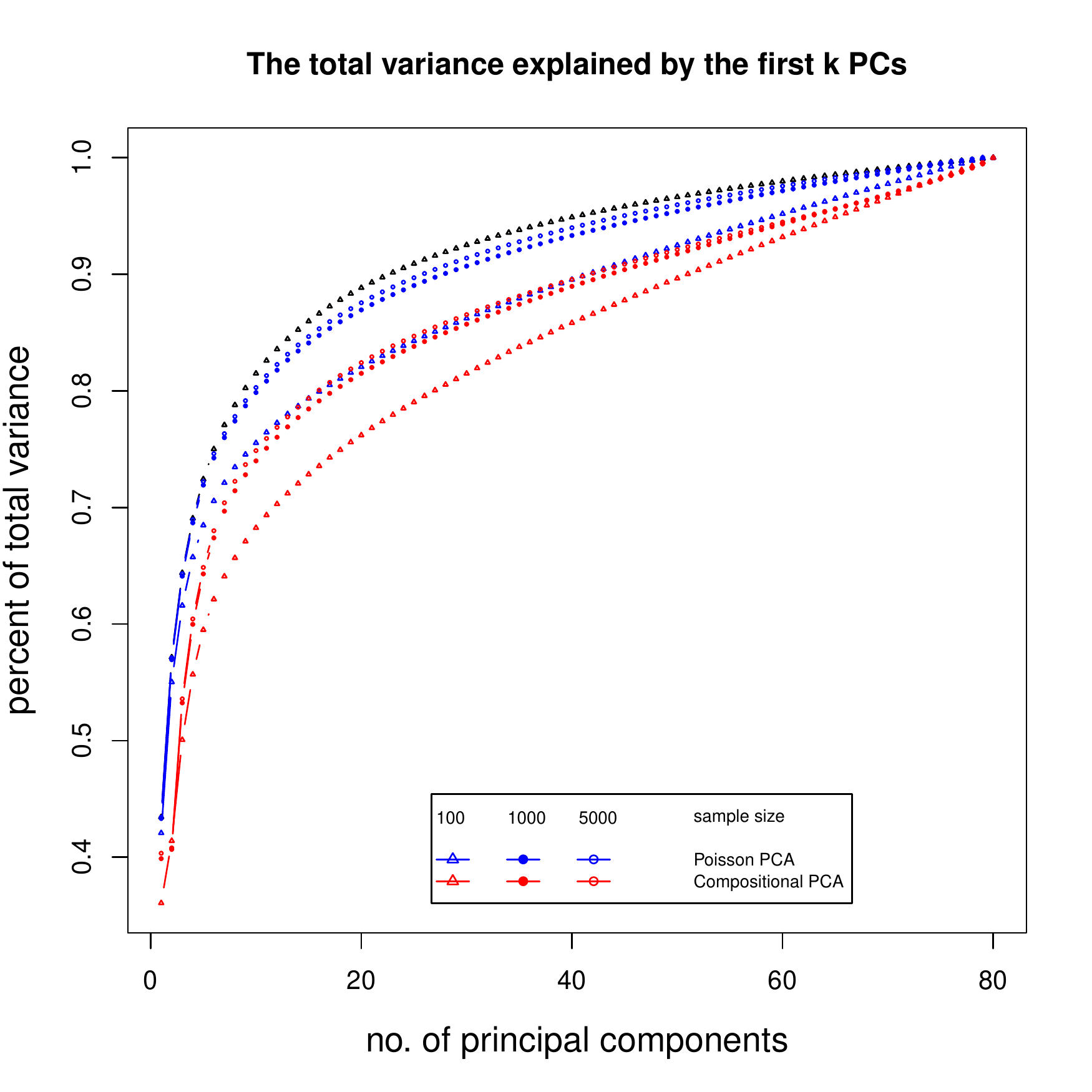}
  \caption{$n$ between 100 and 5,000}
\end{subfigure}
\begin{subfigure}{0.48\textwidth}
  \includegraphics[width=7.5cm]{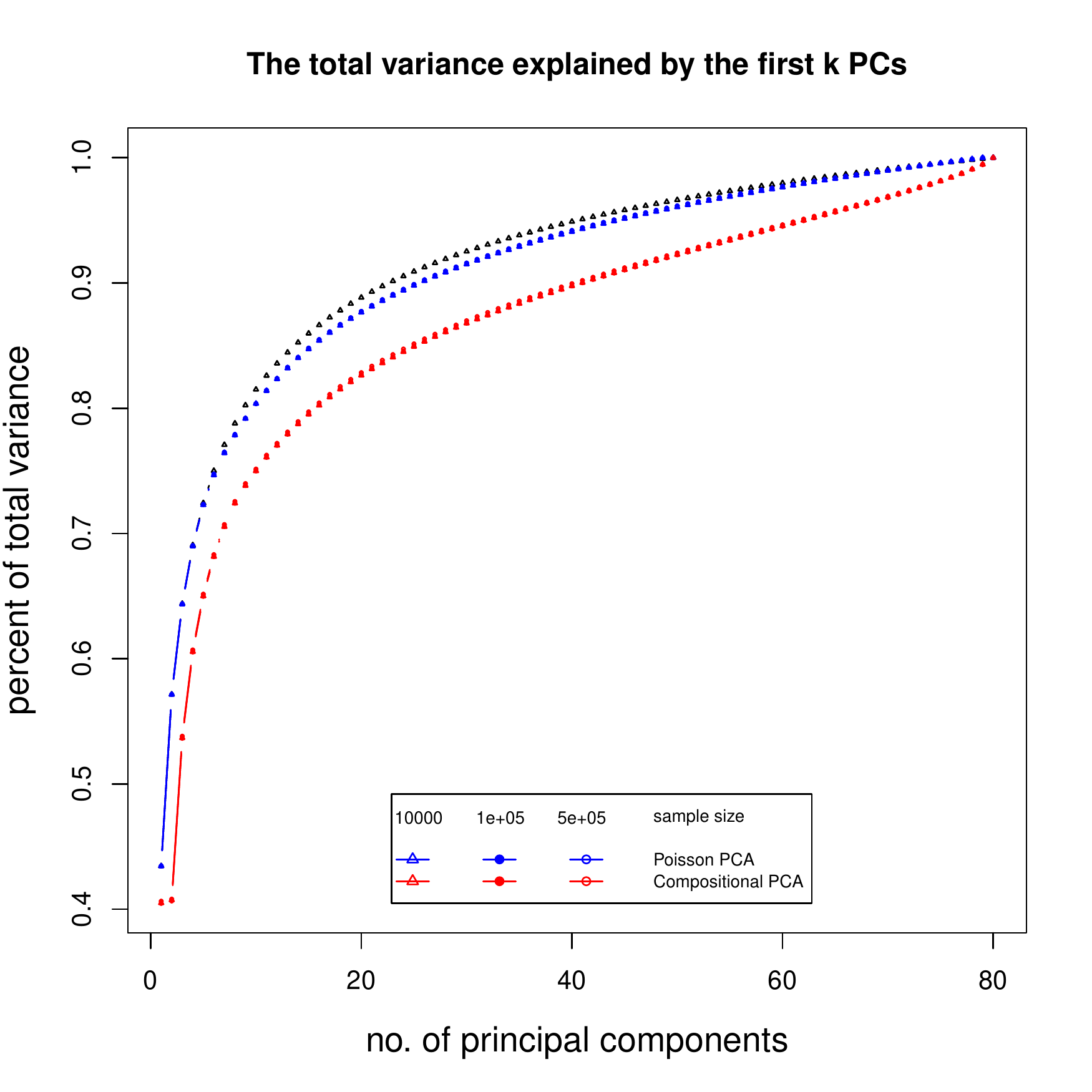}
  \caption{$n$ between $10^4$ and $5\times 10^5$}
\end{subfigure}
  
\end{figure}

Figure~\ref{logsimresultsseqdepth} shows the results for log-normal
simulated data with sequencing depth noise. In the compositional log
Poisson means simulation, we use the log of observed total counts as
an offset in the parametric PLN model. We see that this is a
challenging problem, but that our compositional sequencing depth
correction method (from Section~\ref{CompositionalCovariance}) offers
a slight improvement over the PLNPCA method. We do not include larger
sample sizes because of the computation time needed to fit the PLNPCA
model.  We also compare our two methods for correcting sequencing
depth noise when the latent covariance matrix is not compositional. We
see that the minimum variance approach (from
Section~\ref{MinvarCorrection}) has not done well at identifying the
first eigenvector, because it has not removed enough sequencing depth
noise. However, retaining the signal in the direction 1 allows it to
get improved results for later principal components. Overall it seems
that in most circumstances it would be better to use the compositional
method, but in rare cases where we are aiming to find a large number
of principal components, it may be preferable to use the minimum
variance approach. Additional simulation results are shown in
Supplemental Figure~\ref{logsimresultsseqdepthsupplemental}.

\begin{figure}[htbp]

  \caption{Comparison of log-normal PCA methods with sequencing depth
    noise and Poisson noise}\label{logsimresultsseqdepth} Left:
  compositional log Poisson means, Poisson PCA
  (blue) with compositional sequencing depth correction, compared with
  full rank PLN (green) with total count used as offsets.
  Right: non-compositional log Poisson means. We compare minimum variance
  estimation (blue), compositional sequencing depth
 correction (green) and no sequencing depth correction (red).
  The truth is shown in black on both figures. Results are averaged
  over 100 simulations. $p=80$.

\begin{subfigure}{0.48\textwidth}
  \includegraphics[width=7.5cm]{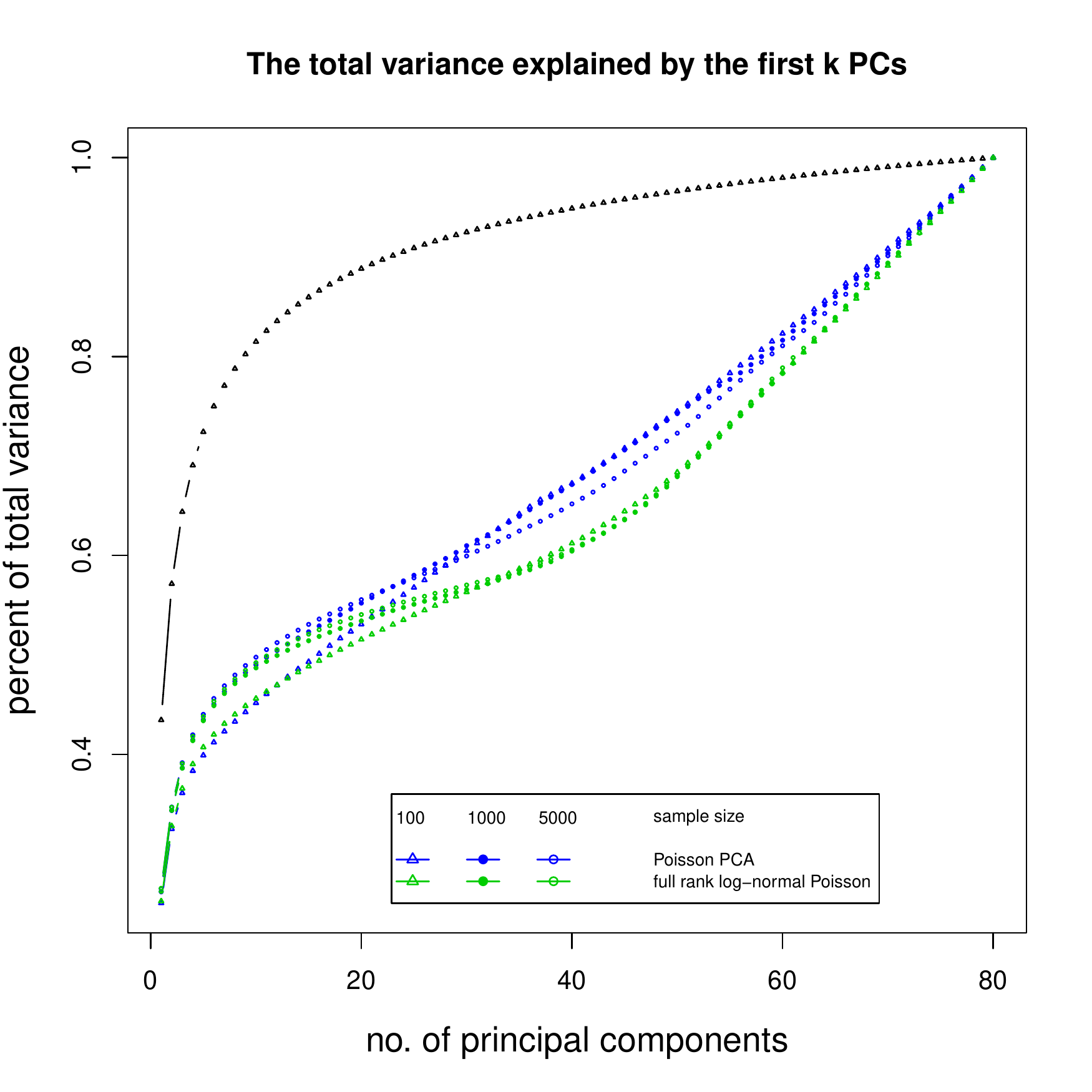}
  \caption{Compositional}
\end{subfigure}
\begin{subfigure}{0.48\textwidth}
  \includegraphics[width=7.5cm]{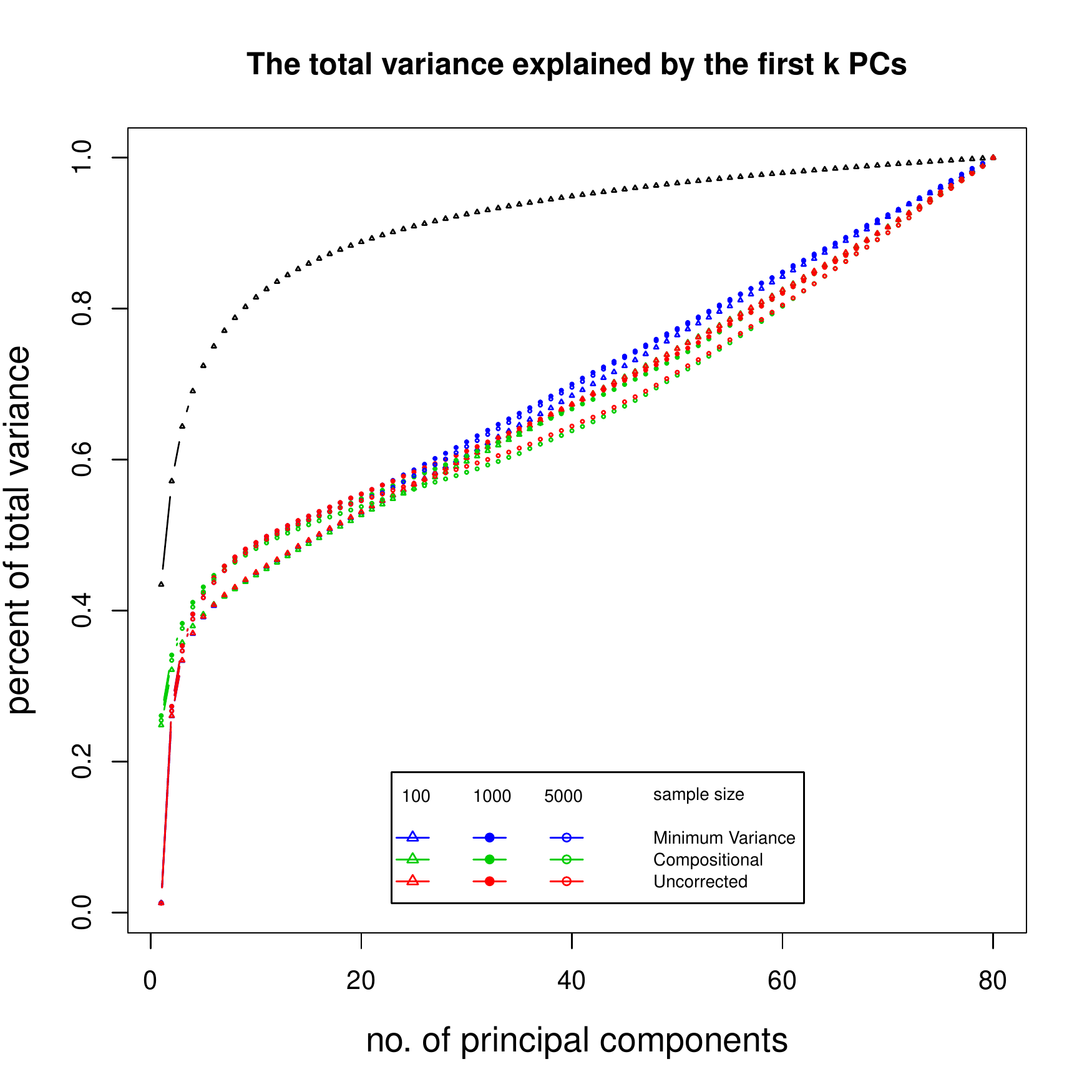}
  \caption{Non-compositional}
\end{subfigure}
  
\end{figure}

\subsection{Projection onto Principal Component Space}

We now test the performance of the projection method from
Section~\ref{Projection} on simulated data. We use our standard
log-normal Poisson simulation without sequencing depth correction. We
compare two projection approaches: (a) pure normal error --- we take
the log of the observed data and project this onto the principal
component space (substituting $-3$ for $\log(0)$); (b) the combined
likelihood method from Section~\ref{Projection}. We apply each of
these projections both using the true covariance matrix, and using the
estimated covariance matrix from Poisson PCA. 

For assessing the quality of our projection, we have the true
transformed Poisson mean $\log(\Lambda)$, so we can project this
Poisson mean onto the principal component space, and measure the
distance between the correct projection and the estimated
projection. Since the projection is orthogonal, Euclidean distance
suffices for this measure.

Figure~\ref{ProjectionSimulationResults} compares the results on both
the true principal component space, and the estimated principal
component space. We see that when we project onto the true principal
component space, our projection method outperforms the na\"ive
projection method for all numbers of principal components. When we
project onto the estimated principal component space, we see that our
method performs excellently for any projection that might ever be used
in practice. However, when we try to project onto high dimensional
component spaces, our method starts to perform less well. In this
case, because there is very little signal in the remaining dimensions,
our method may soak up some of the noise. Optimisation is also more
challenging in higher dimensions, so inaccuracies in the optimisation
may cause deterioration in performance. Since
the typical use of PCA is to obtain a low-dimensional representation
of the data, performance on small numbers of principal components is
most important. Therefore, the results of our method are
good. Similar results for $p=30$ are in Supplementary Figure~\ref{ProjectionSimulationResultsp30}.

\begin{figure}[htbp]

\caption{Comparison of projection methods for
  log-transformed Poisson PCA}\label{ProjectionSimulationResults}

Figures on the top row show MSE of projection onto the true principal
component space (compared to projection of true transformed latent
Poisson means), while figures on the bottom row show MSE of projection
onto the estimated principal component space (compared to projection
of true transformed latent Poisson means). The black line is our
projection method, while the red line is for the na\"ive transform then
project method. Dimension is $p=80$.

\begin{subfigure}{0.32\textwidth}
  \includegraphics[width=5cm]{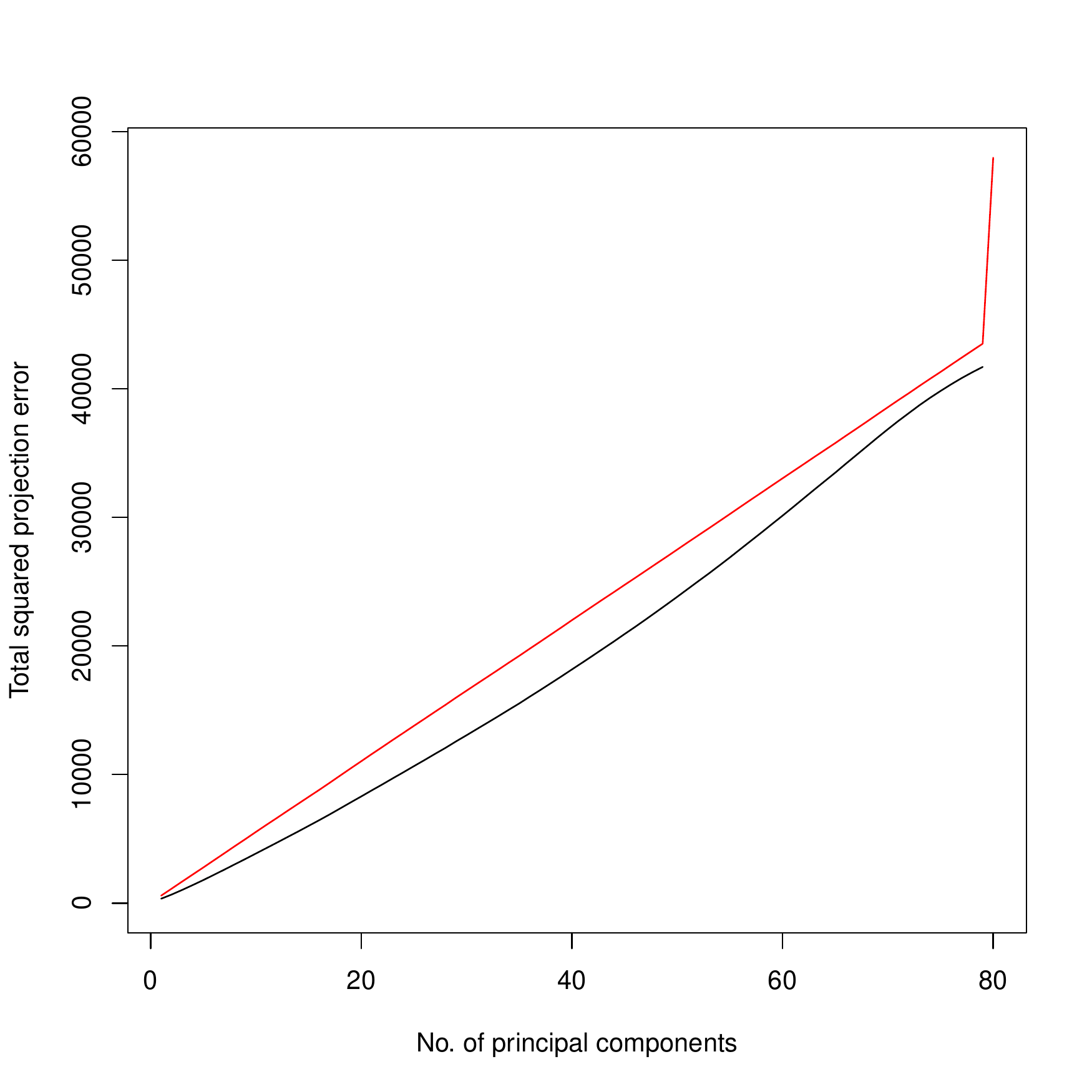}
  \caption{$n=100$, true}
\end{subfigure}
\begin{subfigure}{0.32\textwidth}
  \includegraphics[width=5cm]{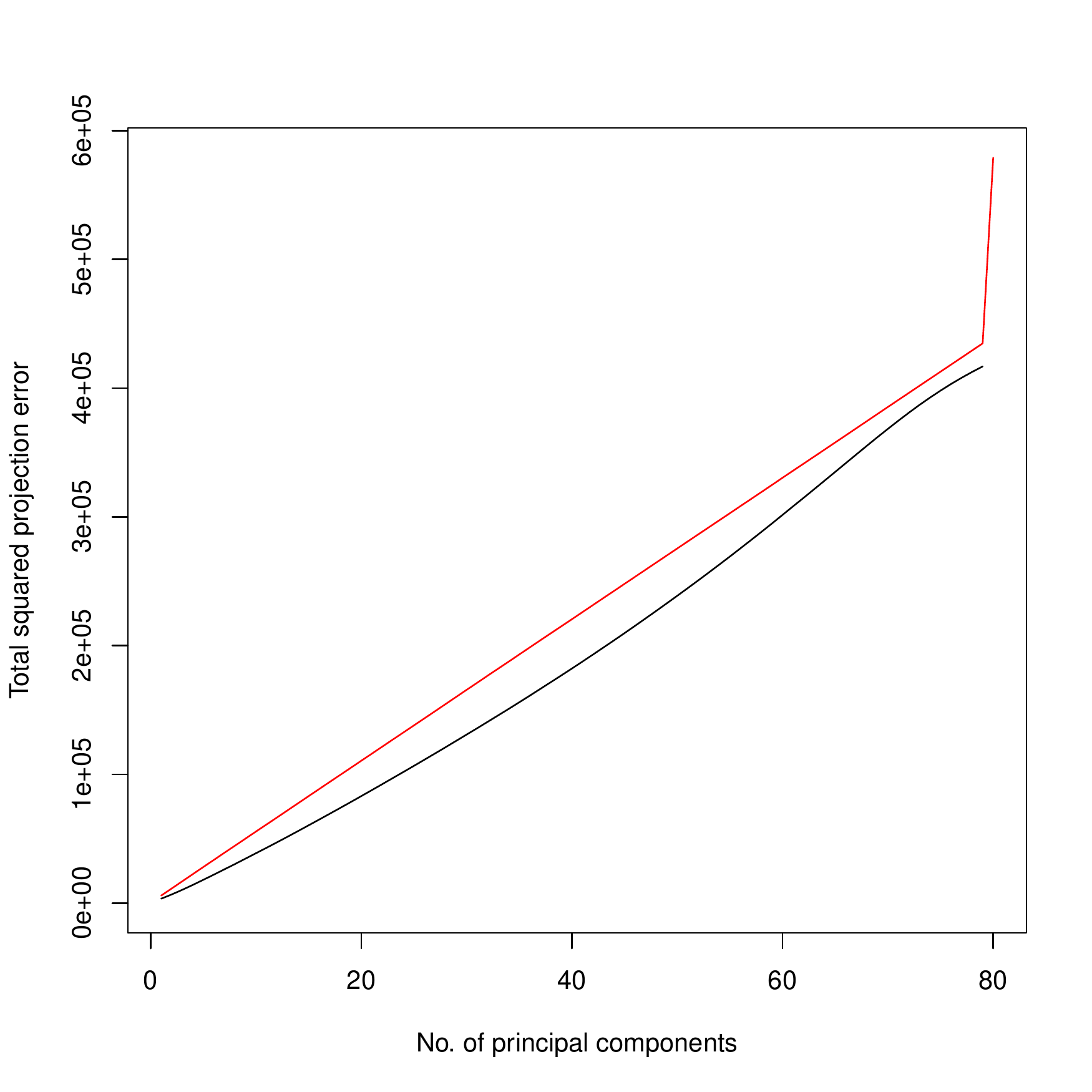}
  \caption{$n=1,000$, true}
\end{subfigure}
\begin{subfigure}{0.32\textwidth}
  \includegraphics[width=5cm]{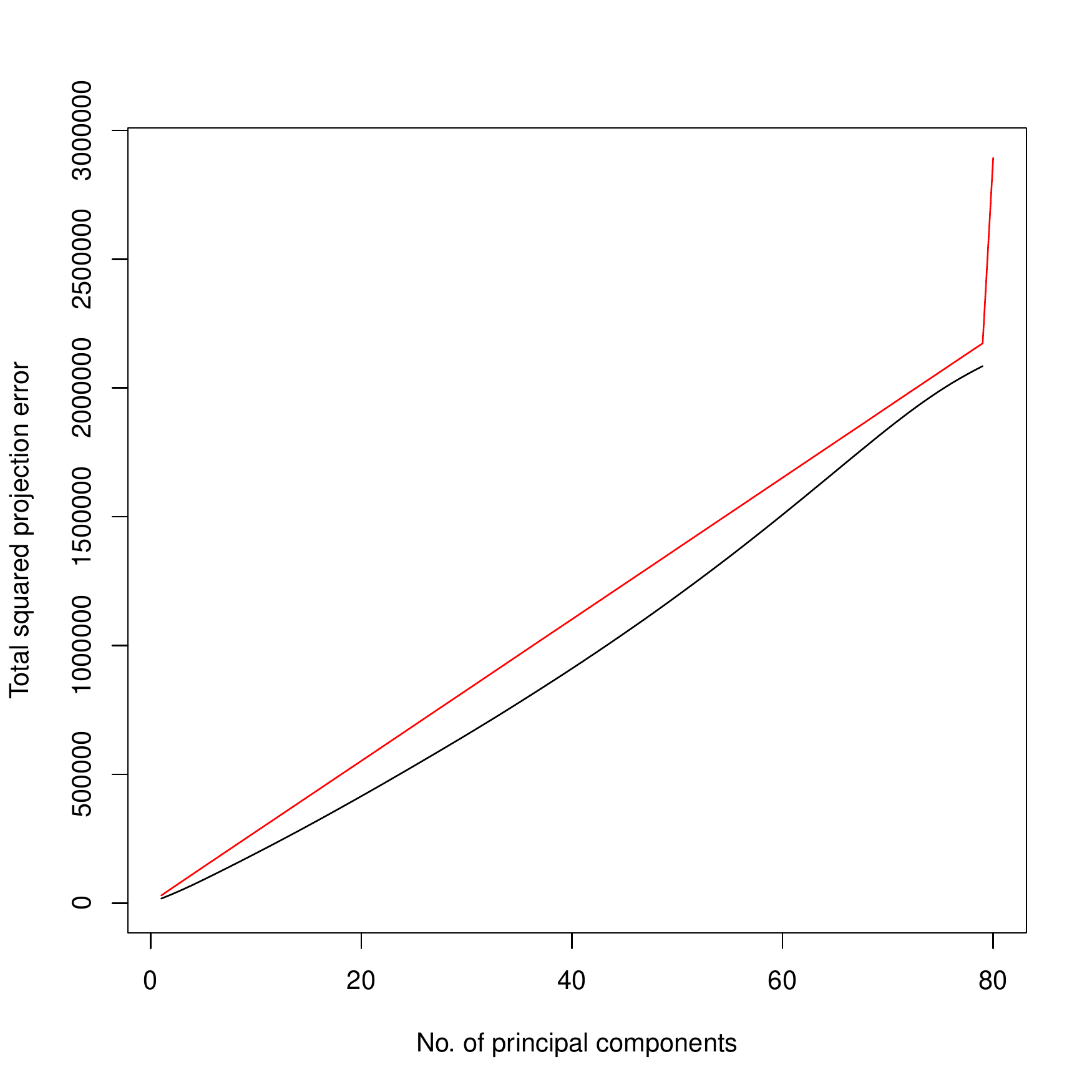}
  \caption{$n=5,000$, true}
\end{subfigure}

\begin{subfigure}{0.32\textwidth}
  \includegraphics[width=5cm]{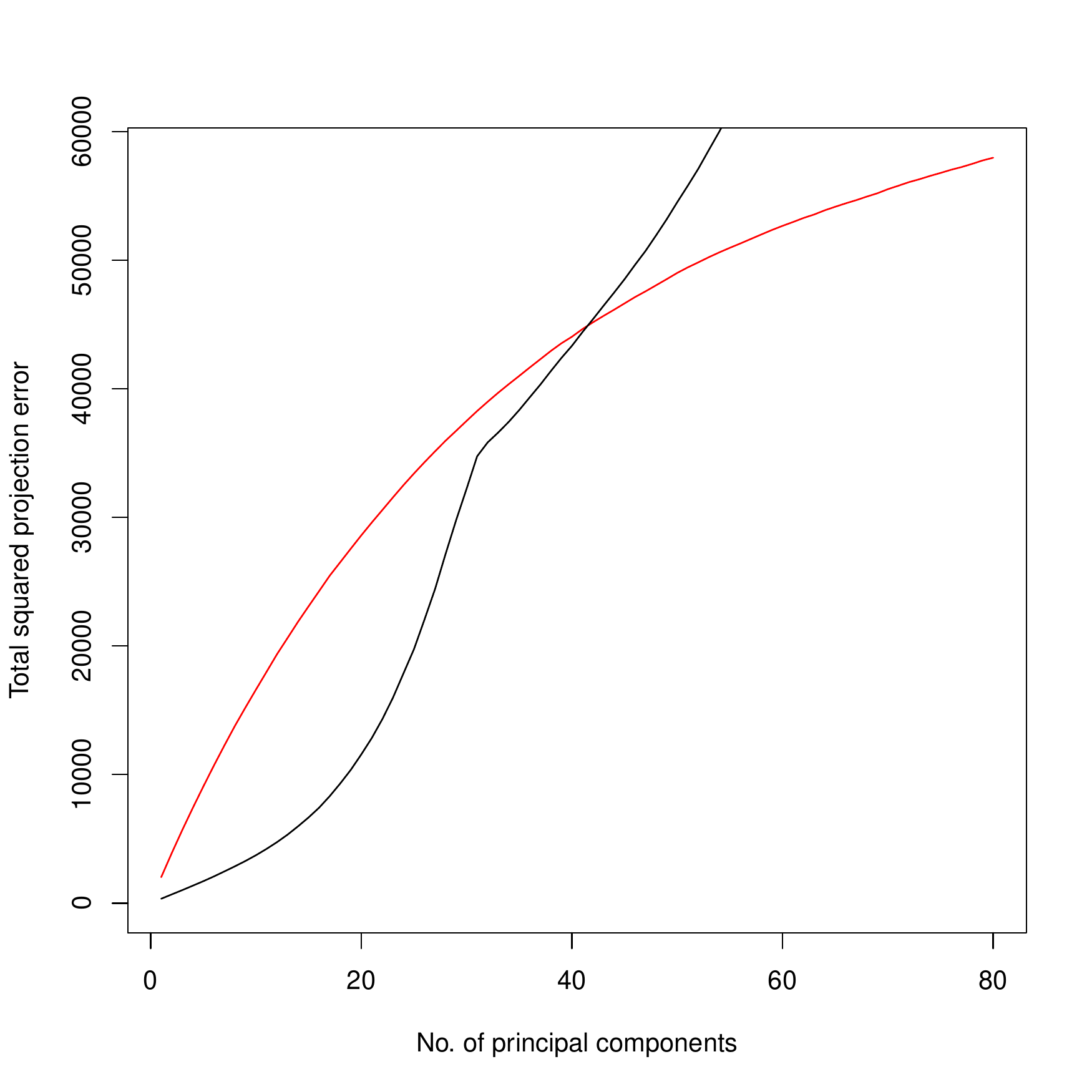}
  \caption{$n=100$, estimate}
\end{subfigure}
\begin{subfigure}{0.32\textwidth}
  \includegraphics[width=5cm]{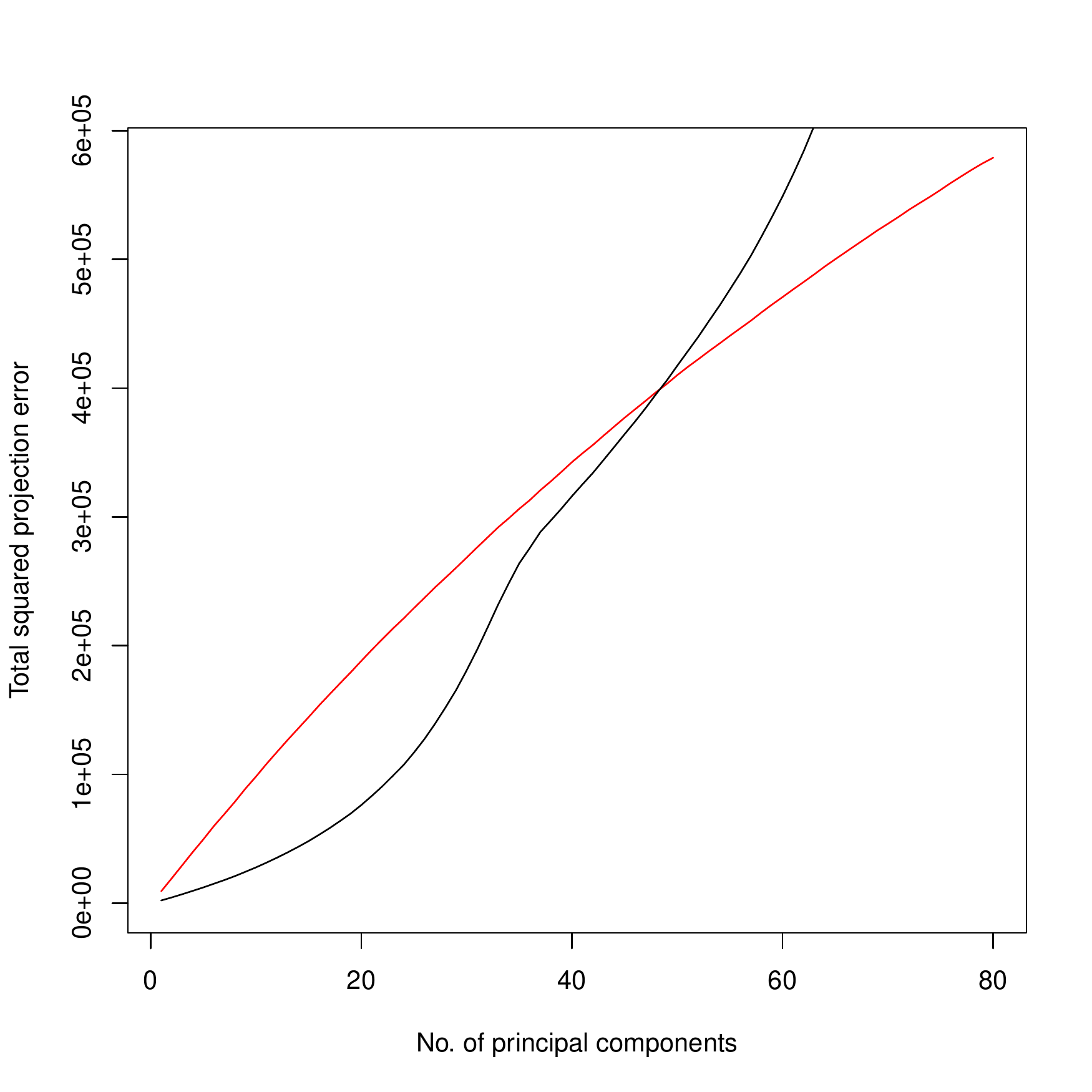}
  \caption{$n=1,000$, estimate}
\end{subfigure}
\begin{subfigure}{0.32\textwidth}
  \includegraphics[width=5cm]{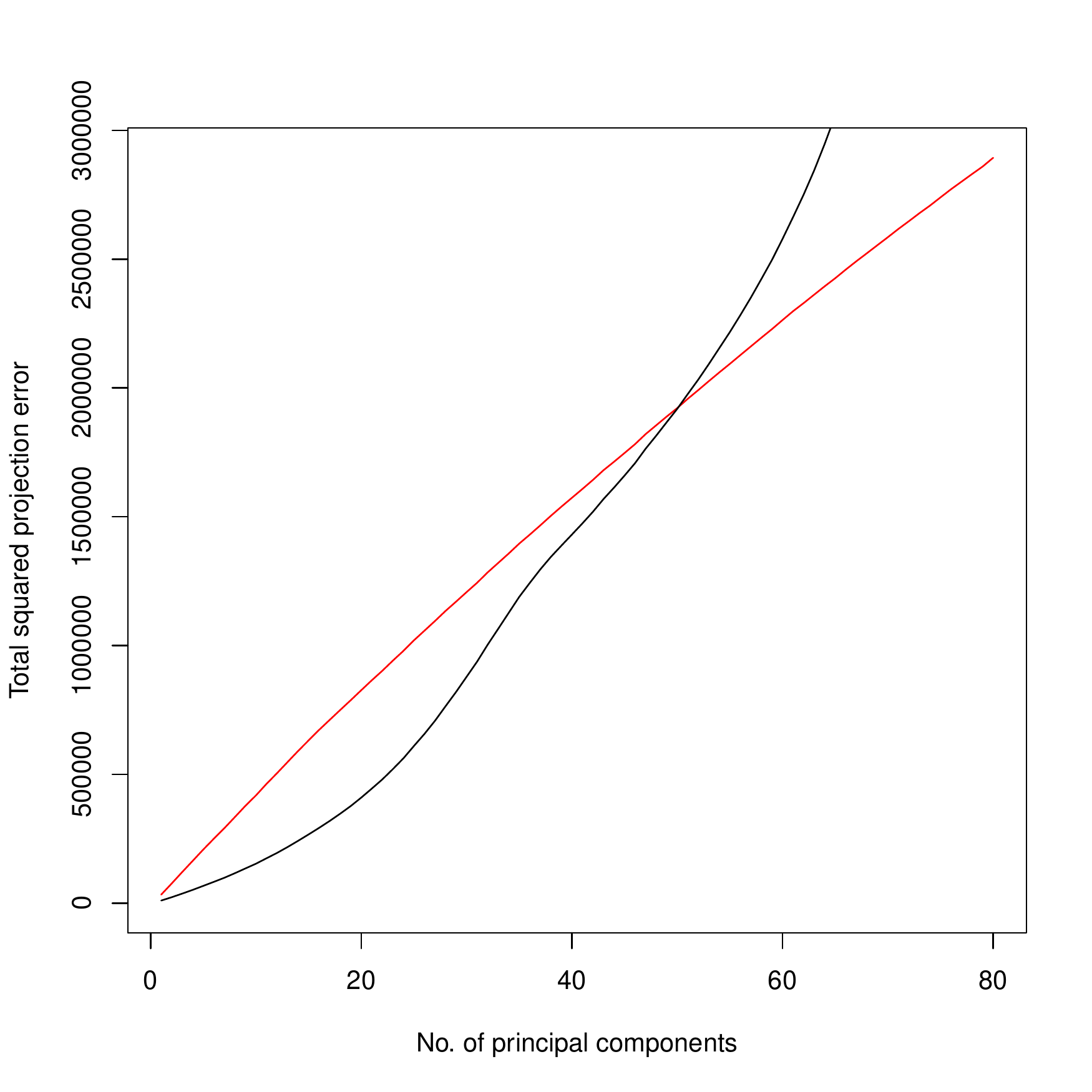}
  \caption{$n=5,000$, estimate}
\end{subfigure}

\end{figure}

\section{Real Data Analysis}\label{RealData}

\subsection{The Moving Picture Dataset}

The moving picture dataset (Caporaso {\em et al.} 2011) is a large
temporal microbiome dataset comprising samples from 4 body sites of
two individuals over a period of approximately two years. There are
over 300 samples from each site for individual~2 and over 100 samples
from each site for individual~1. Previous analysis of this dataset has
shown that it is relatively easy to distinguish the two individuals
from the species-level OTUs, but using higher-level OTUs to
distinguish them is more difficult. We apply our Poisson PCA method
across a range of levels from phylum to genus from this
dataset. Besides making the problem challenging compared to the
results of previous studies, the focus on higher level OTUs also
ensures that the number of variables is less than the number of
observations, which is a necessary condition for the consistency of
classical PCA. We present the results for the first two principal
components at the order level (225 variables) in this section. Results
at other levels are in the supplemenary materials.

There is substantial variation in sequencing depth between samples,
with number of total reads varying between less than 10,000 and over
60,000. The distributions of sequencing depths were different for the
two individuals. Since sequencing depth is largely considered noise
due to experimental procedure, a method which separates samples based
on sequencing depth is not a good method. To prevent methods that fail
to correctly compensate for sequencing depth from wrongly seeming
good, we used subsampling to ensure that the sequencing depths had
the same distribution for each individual.  (Details of the
subsampling procedure are in Supplemental
Appendix~\ref{RealDataSubsamplingProcedure}.) We retain a large
variability of sequencing depth within each individual's samples, but
remove the ability to use sequencing depth as a distinguishing
feature, so that methods that better correct for sequencing depth
noise are expected to perform better on the subsampled data.

\subsection{Data Analysis}

We compare three methods for dimension reduction of the
log-transformed data: our compositional sequencing depth corrected
Poisson PCA method using the projection method from
Section~\ref{Projection}; PLNPCA with rank chosen by BIC in the range
1--30, using total count as an offset; and log compositional PCA,
i.e. naive PCA applied to the log of compositional data with zero
counts replaced by 0.01 before converting to compositional data
(similar results are obtained using other values to replace 0). Since
the differences between the two individuals' microbial communities are
likely to be a large source of variability in the data, we expect to
see a separation between samples from the two individuals after
performing PCA on the data from each body site.

Results for the gut data are in Figure~\ref{MovPicGutOrder}. We see
that all three methods show a reasonable separation between the two
individuals. PLNPCA shows several outliers, which can be explained by
the difficulty of attempting to fit data to a prescribed parametric
distribution if the distribution does not fit.

\begin{figure}[htbp]

\caption{First two principal components: Gut data grouped at order level}\label{MovPicGutOrder}

\begin{subfigure}{0.32\textwidth}
  \includegraphics[width=5cm]{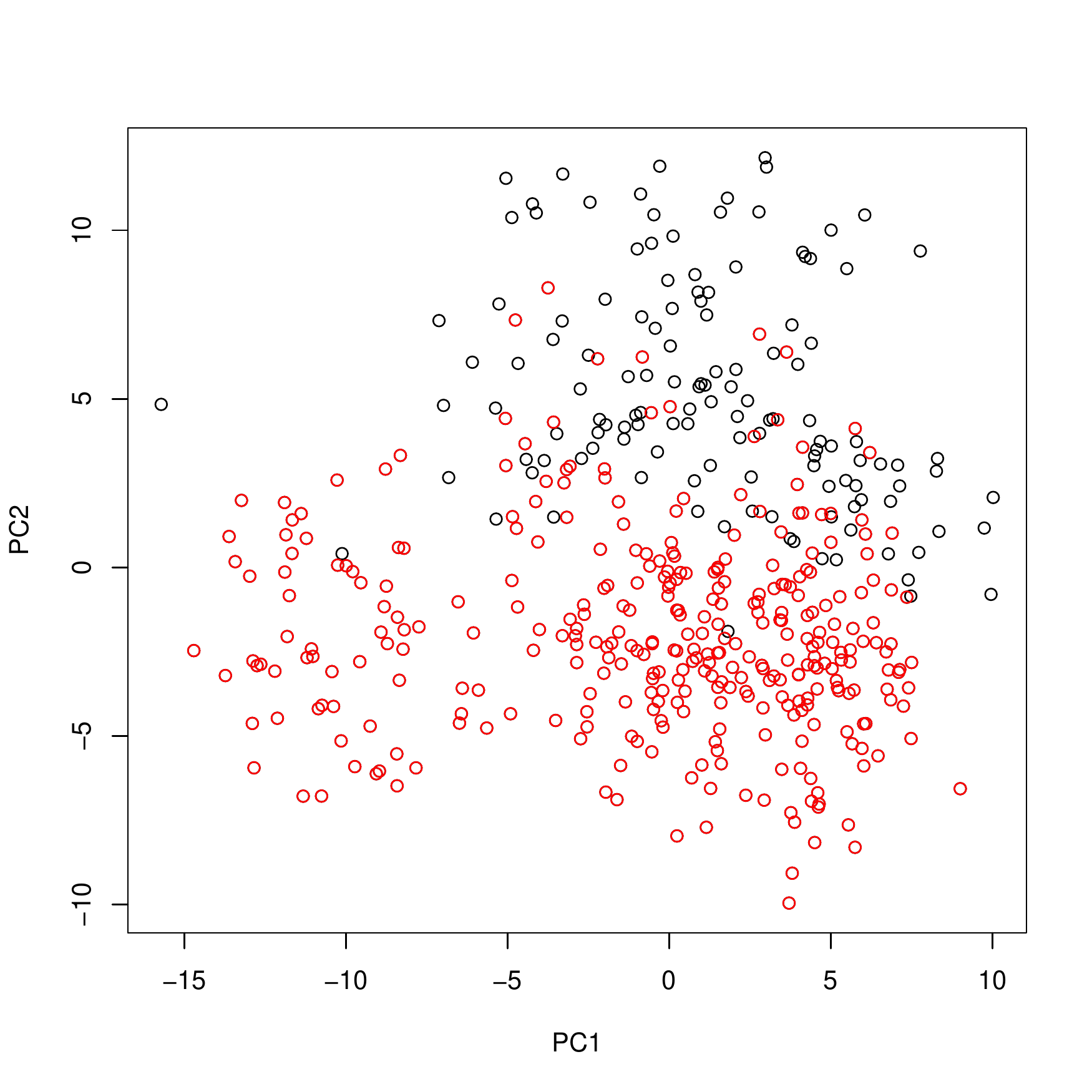}
  \caption{log compositional PCA}
\end{subfigure}
\begin{subfigure}{0.32\textwidth}
  \includegraphics[width=5cm]{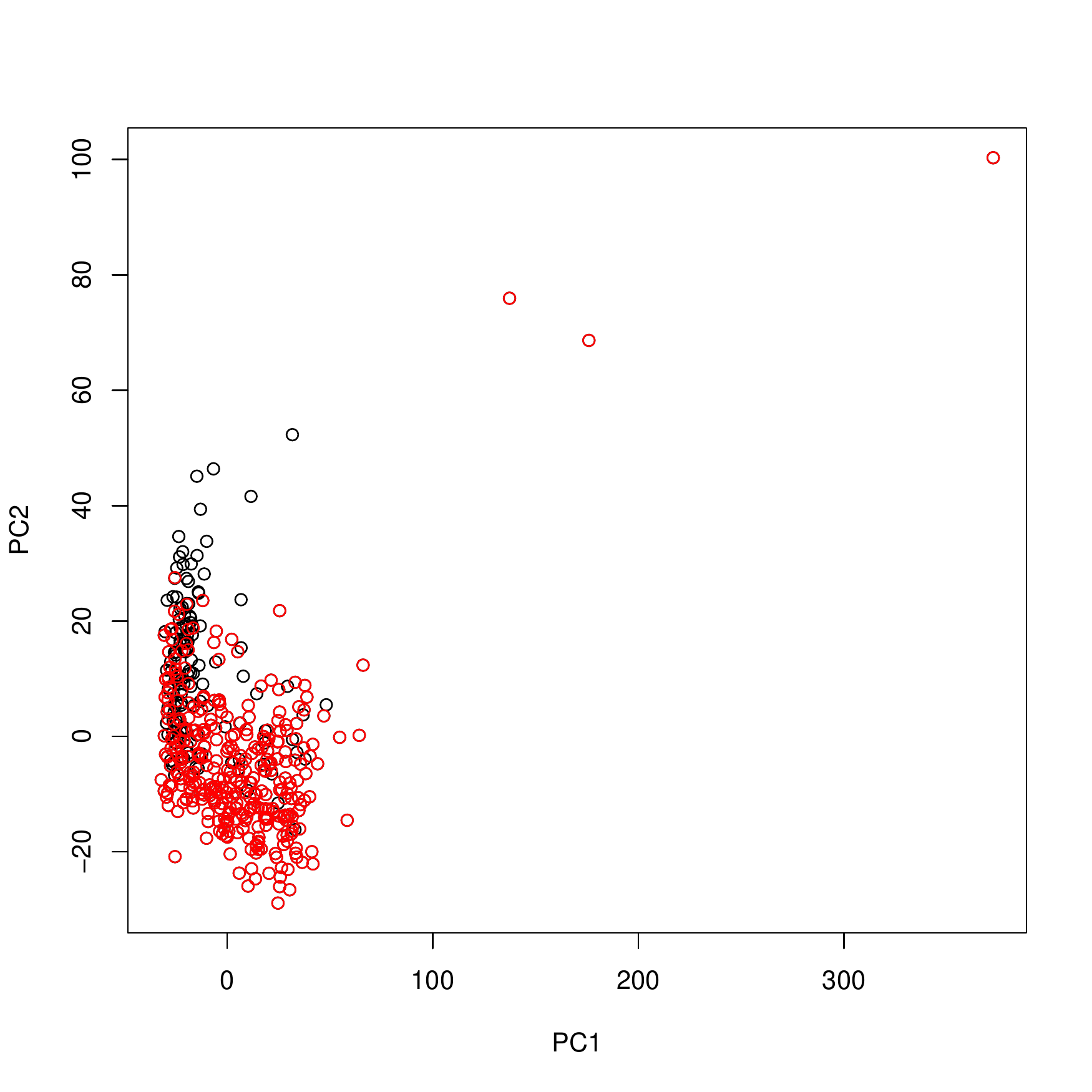}
  \caption{PLNPCA}
\end{subfigure}
\begin{subfigure}{0.32\textwidth}
  \includegraphics[width=5cm]{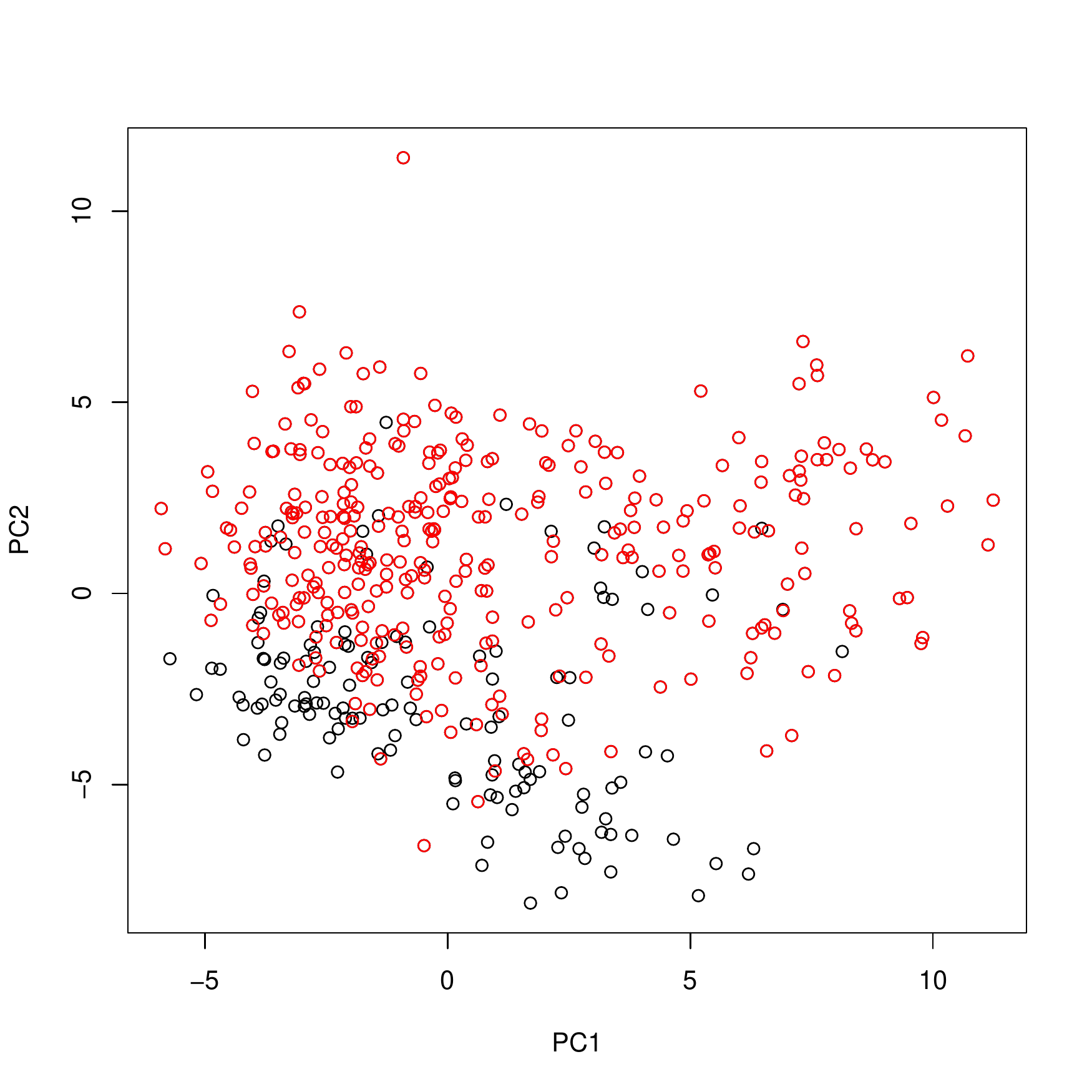}
  \caption{Poisson PCA}
\end{subfigure}

\end{figure}

Figure~\ref{MovPicTongueOrder} shows the projection of the tongue
data. For this dataset, Poisson PCA shows a better separation between
two individuals. PLNPCA has again been heavily influenced by a few
outliers. Log compositional PCA shows a clear clustering, but the
clustering is not related to the individual. Further examination
reveals that the first principal component for both individuals is
heavily influenced by a single order from phylum Cyanobacteria (see
Supplementary Figure~\ref{Cyanobacteria}). This applies to
both Poisson PCA and log compositional PCA. However, it is more
obvious in the case of log compositional PCA, and there is a clear gap
between cases where the observed count of cyanobacteria is zero and where it is
non-zero, which leads to the two clusters seen in
Figure~\ref{MovPicTongueOrder}(a). The abundance of Cyanobacteria does
not vary much between the two individuals, so because this abundance
dominates the first principal component for log compositional PCA, we
do not distinguish between the individuals using this method. Because
Poisson PCA does not attach so much weight to the zero values, it is
able to identify some signal, even with the Cyanobacteria abundance
dominating the first principal component.

\begin{figure}[htbp]

\caption{First two principal components: Tongue data grouped at order level}\label{MovPicTongueOrder}

\begin{subfigure}{0.32\textwidth}
  \includegraphics[width=5cm]{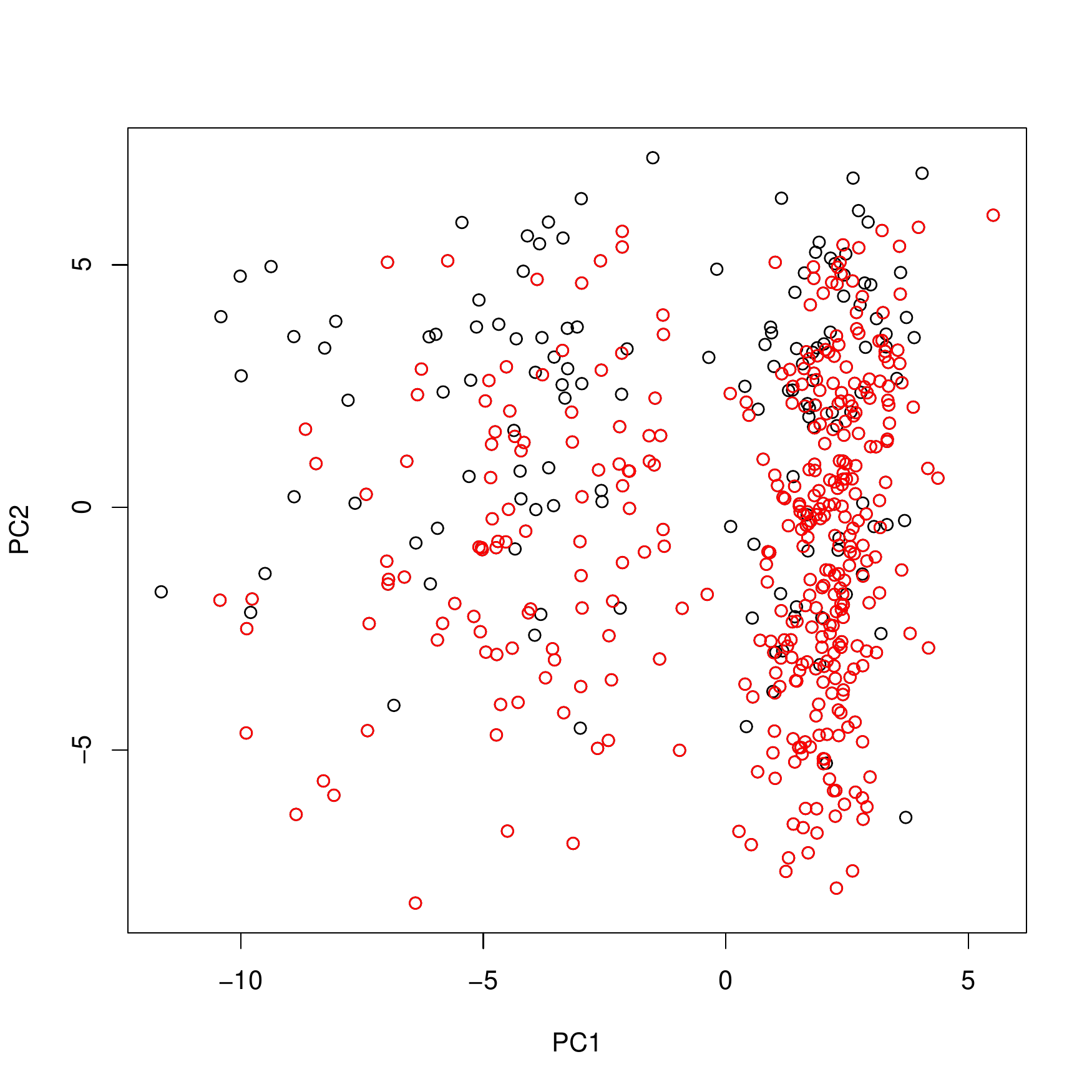}
  \caption{log compositional PCA}
\end{subfigure}
\begin{subfigure}{0.32\textwidth}
  \includegraphics[width=5cm]{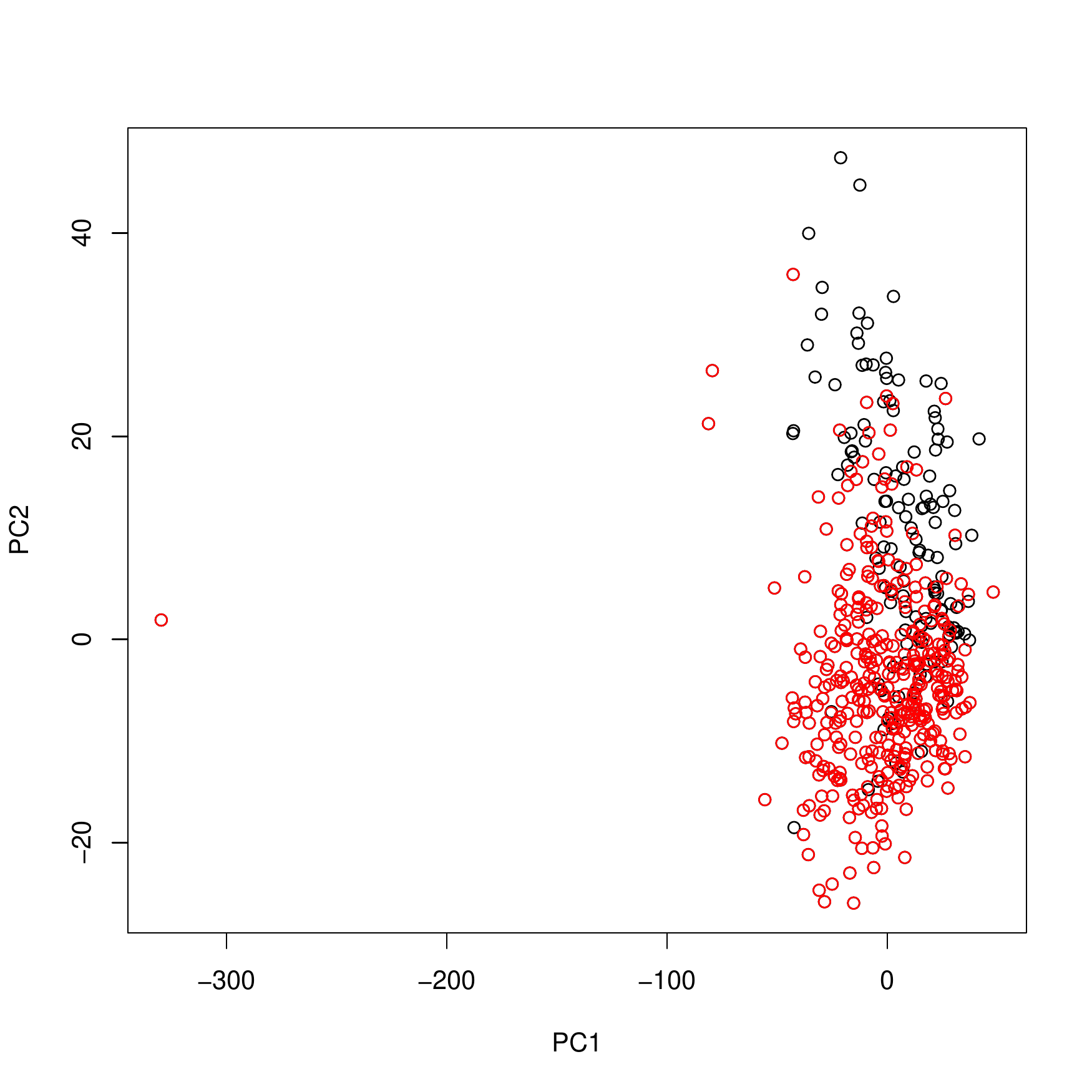}
  \caption{PLN}
\end{subfigure}
\begin{subfigure}{0.32\textwidth}
  \includegraphics[width=5cm]{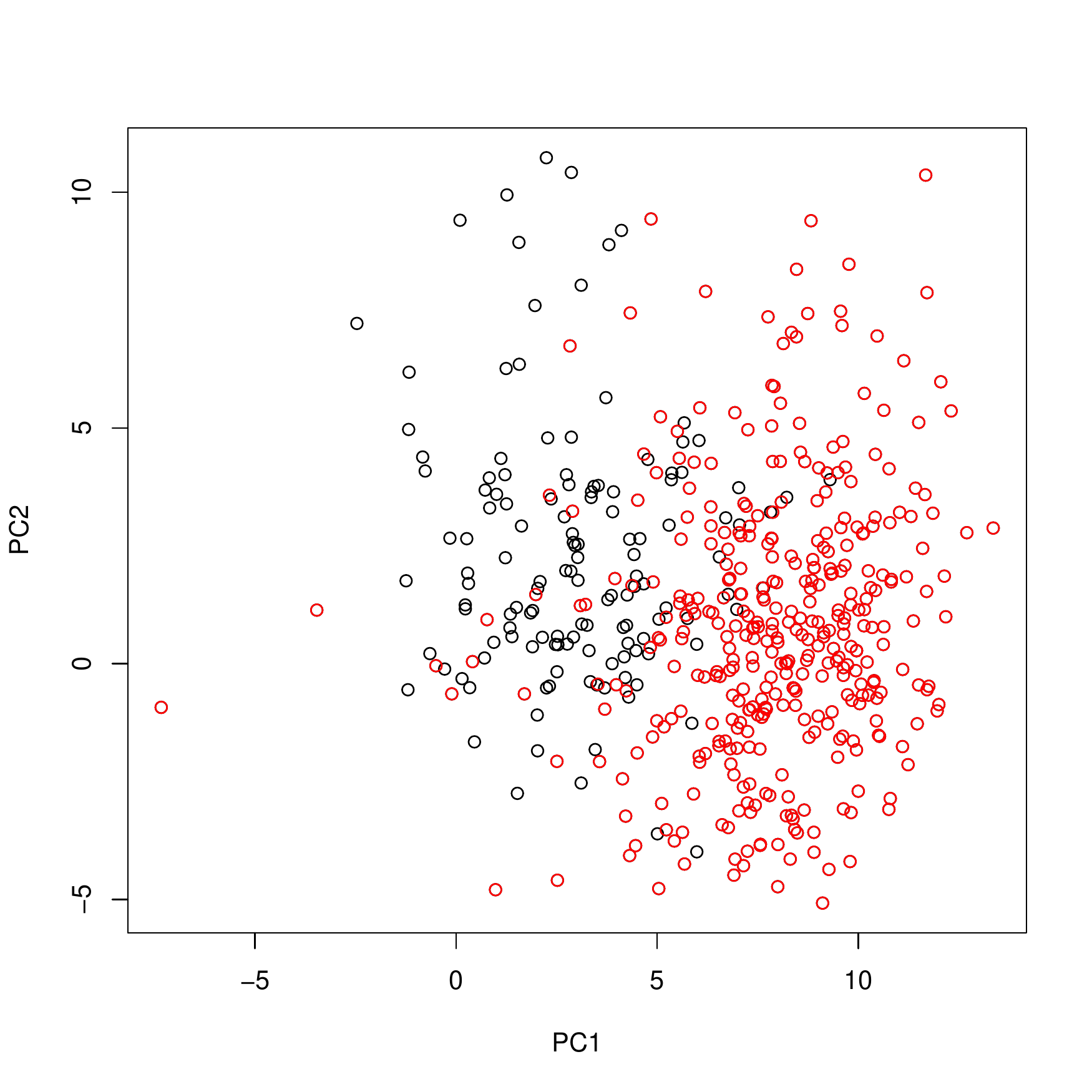}
  \caption{Poisson PCA}
\end{subfigure}

\end{figure}

Figure~\ref{MovPicLPalmOrder} shows the results for the left
palm. Previous research (e.g. Shafiei {\em et al.} 2015, Cai {\em et al.}
2016) has found these samples harder to classify than the gut or
tongue data, so it is not surprising that all PCA methods show more
mixing of the two groups. However, we see the same patterns as for the
tongue data, with Poisson PCA producing the best separation, PLNPCA
being adversely influenced by outliers and log compositional PCA
producing a false clustering. This time, the false clustering can be
seen to be closely related to sequencing depth (see Supplementary
Figure~\ref{MovPicLPalmOrderSeqDepth}).

\begin{figure}[htbp]

\caption{First two principal components: Left palm data grouped at order level}\label{MovPicLPalmOrder}

\begin{subfigure}{0.32\textwidth}
  \includegraphics[width=5cm]{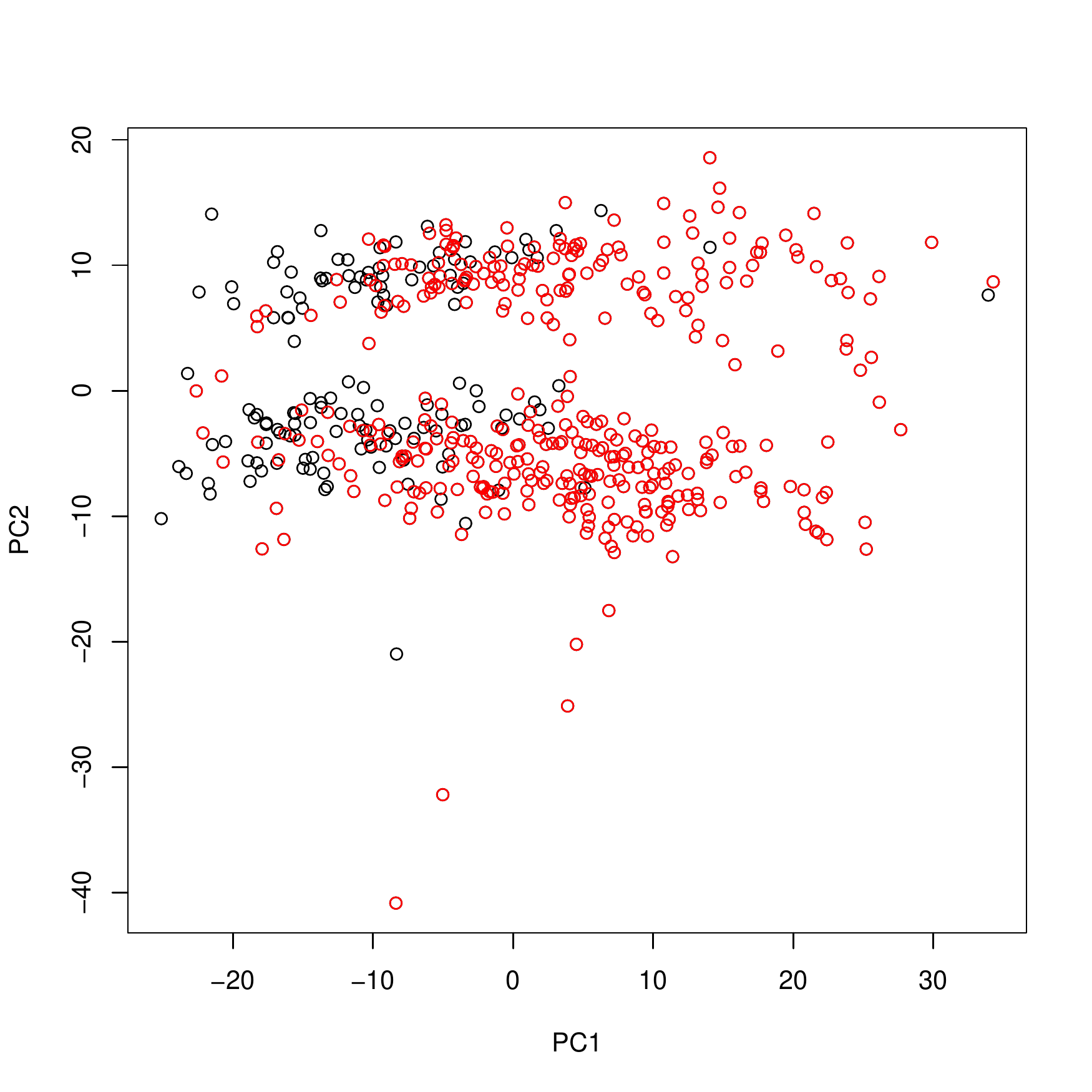}
  \caption{log compositional PCA}
\end{subfigure}
\begin{subfigure}{0.32\textwidth}
  \includegraphics[width=5cm]{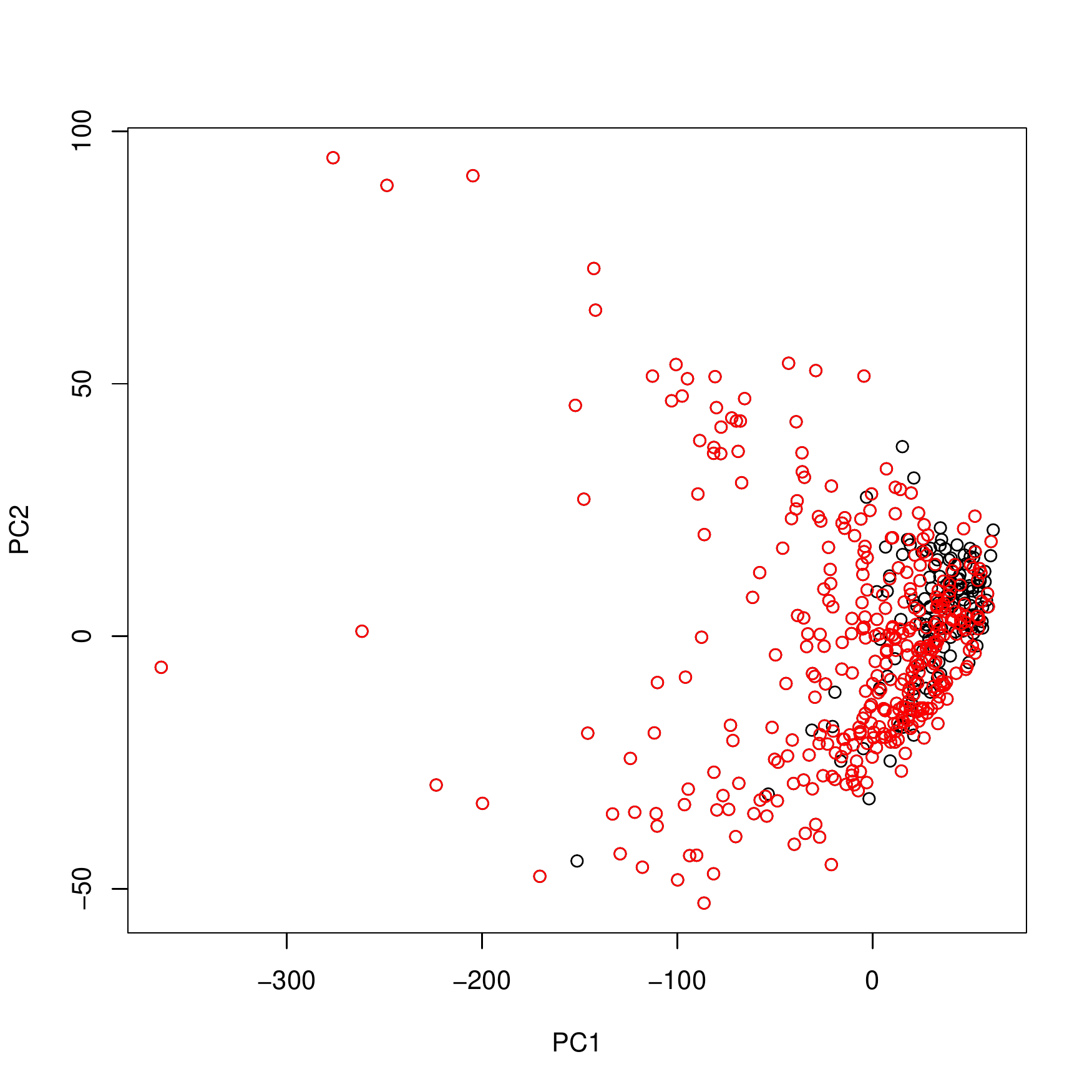}
  \caption{PLN}
\end{subfigure}
\begin{subfigure}{0.32\textwidth}
  \includegraphics[width=5cm]{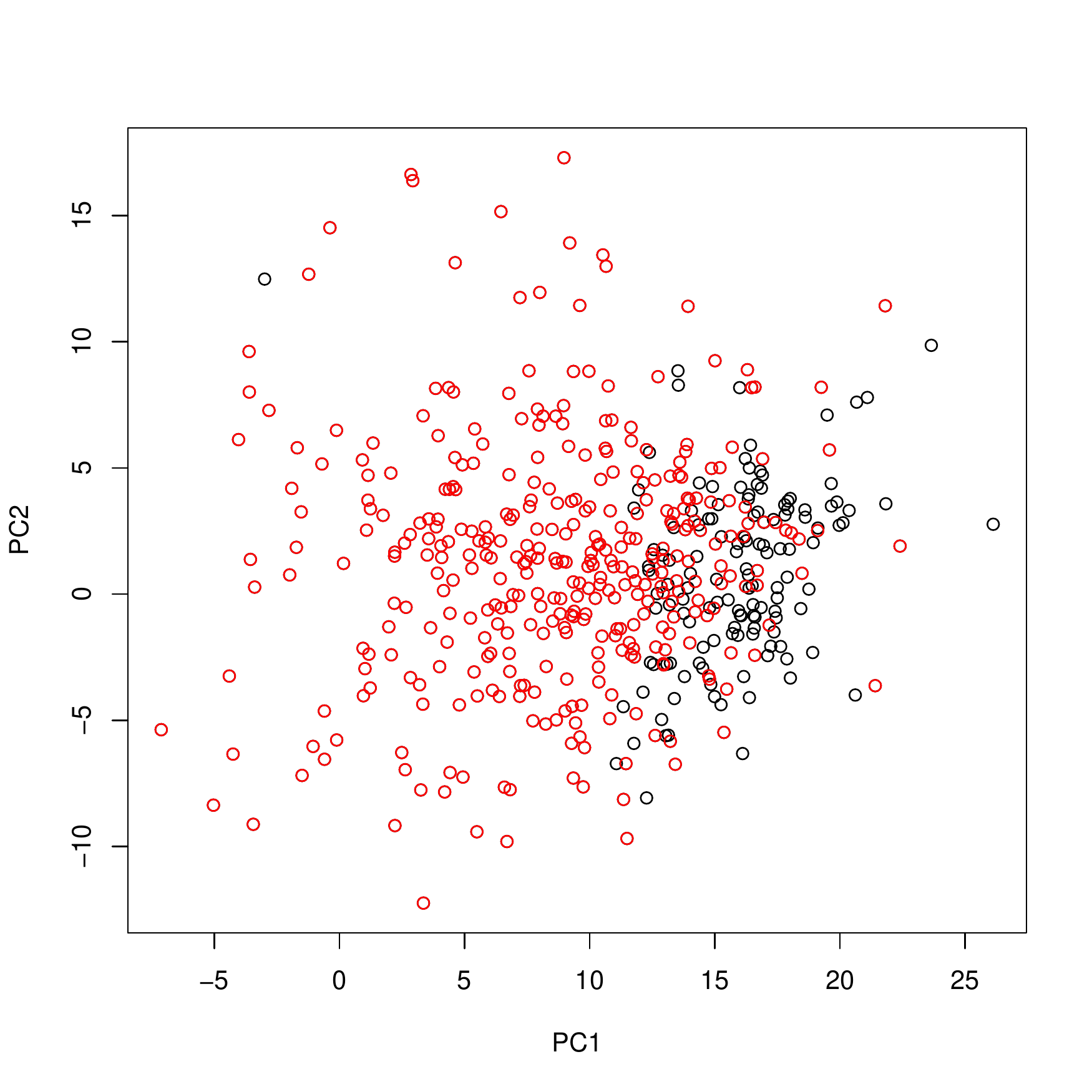}
  \caption{Poisson PCA}
\end{subfigure}

\end{figure}

Figure~\ref{MovPicRPalmOrder} shows the results for the right
palm. Previous studies have found the right palms the most difficult
locations to distinguish between the individuals, and the PCA results
confirm this, with substantial mixing of the groups using all PCA
methods, with log compositional PCA perhaps showing the best
separation (certainly the best linear separation).

\begin{figure}[htbp]

\caption{First two principal components: Right Palm data grouped at order level}\label{MovPicRPalmOrder}

\begin{subfigure}{0.32\textwidth}
  \includegraphics[width=5cm]{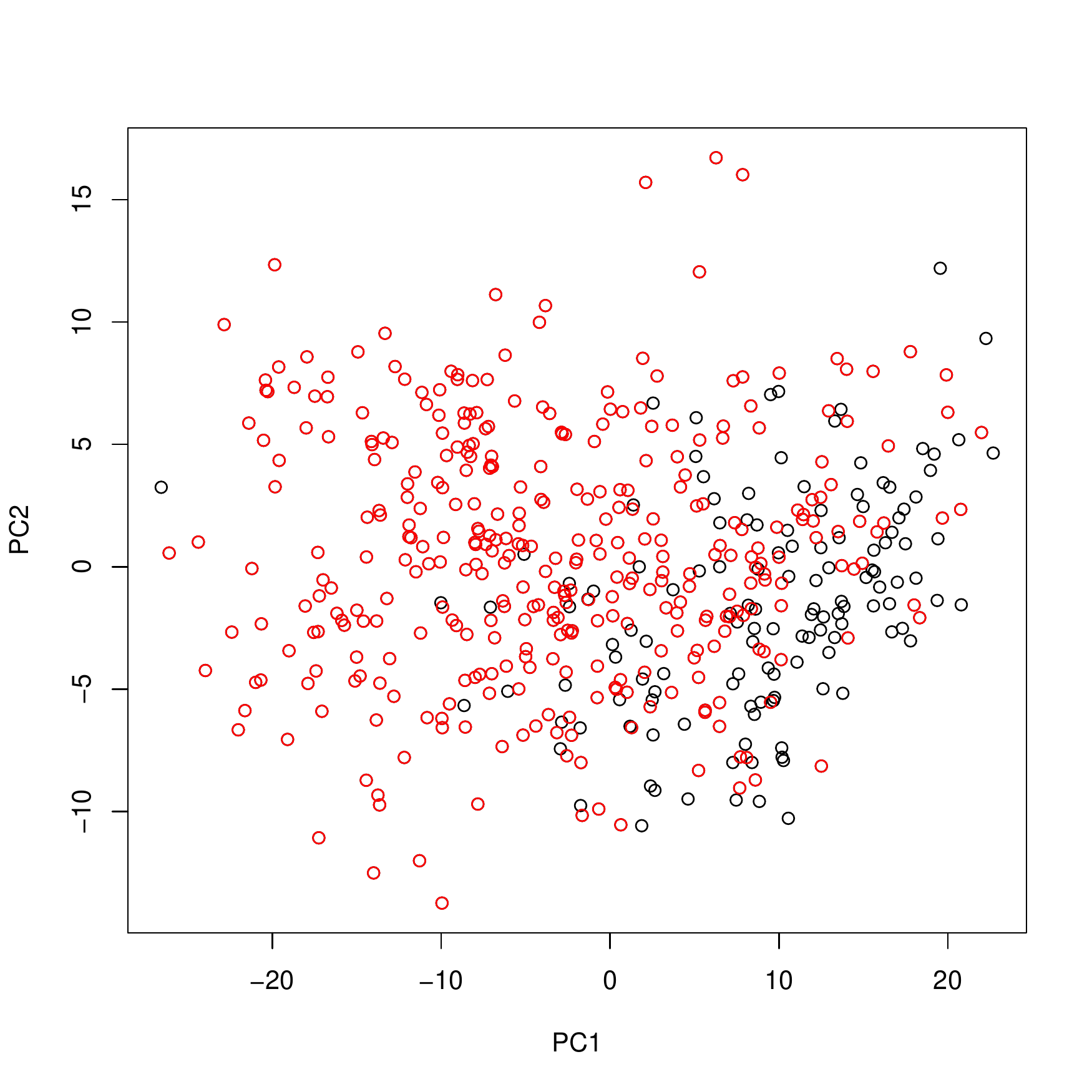}
  \caption{log compositional PCA}
\end{subfigure}
\begin{subfigure}{0.32\textwidth}
  \includegraphics[width=5cm]{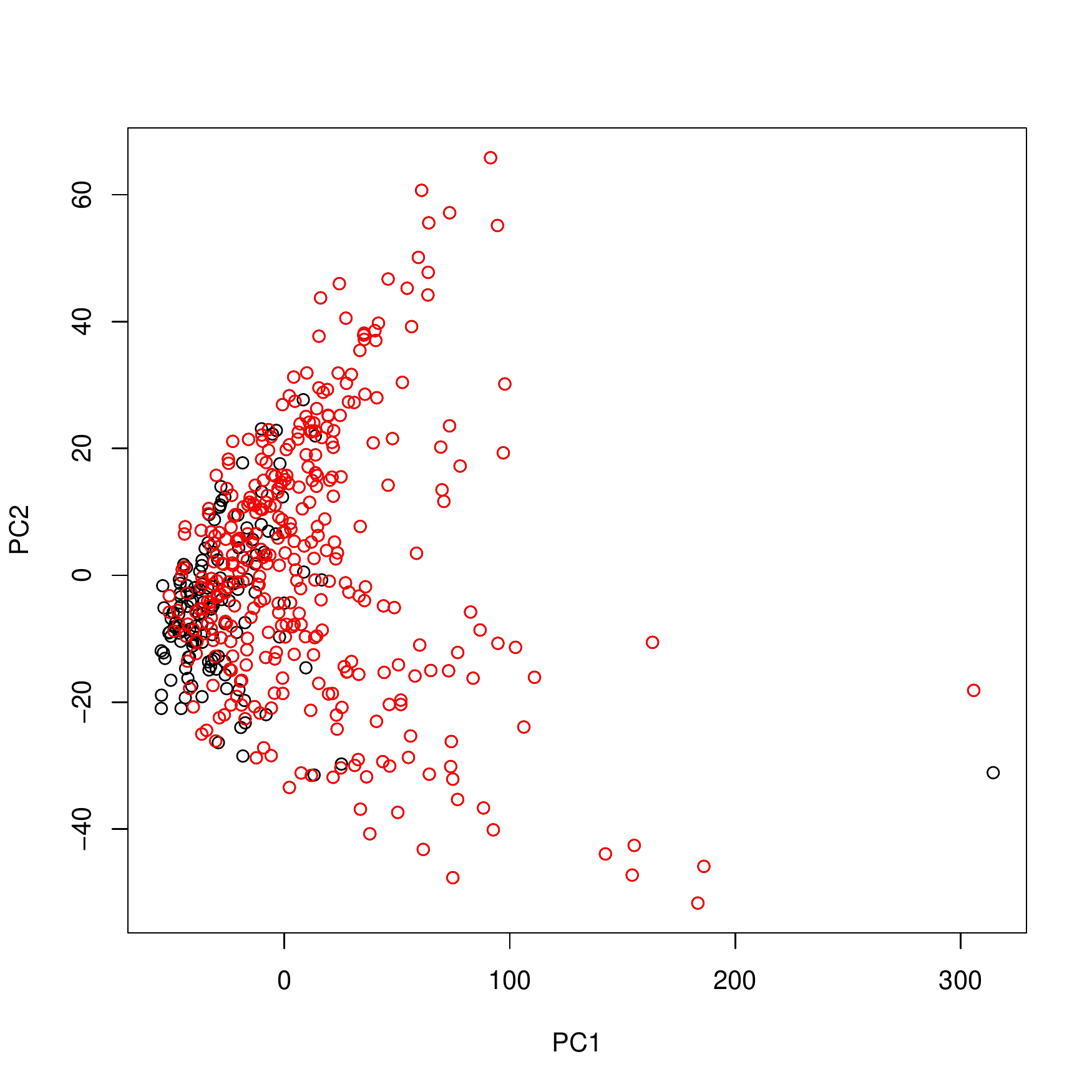}
  \caption{PLN}
\end{subfigure}
\begin{subfigure}{0.32\textwidth}
  \includegraphics[width=5cm]{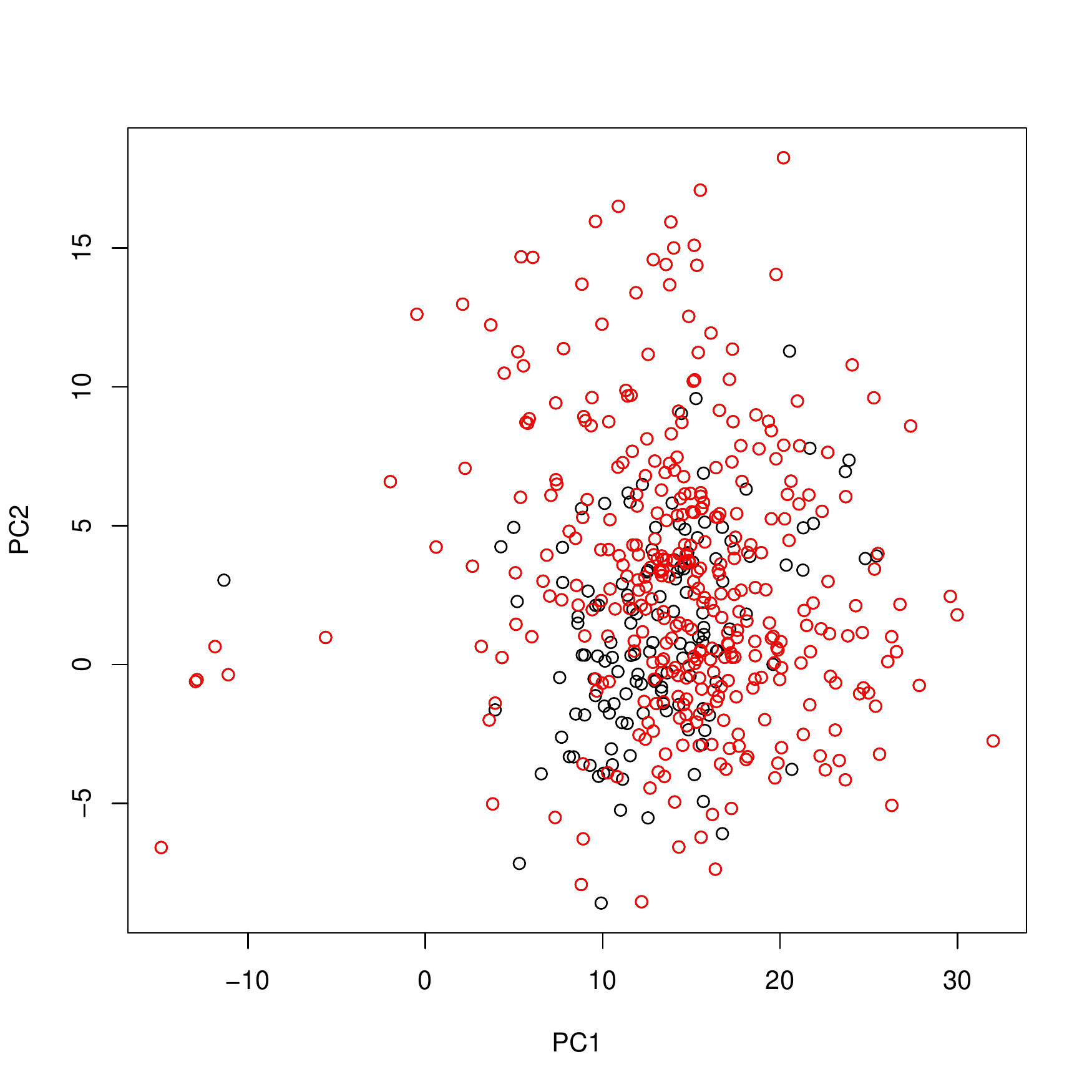}
  \caption{Poisson PCA}
\end{subfigure}

\end{figure}

\section{Conclusions and Future Work}\label{Conclusions}

\subsection{Conclusions}

We have developed methods for correcting for Poisson measurement error
in estimating the variance of a latent variable. Our methods can
estimate the variance of the latent Poisson means, or the variance of
some transformation of the latent means. We have also incorporated
ways to adjust our methods to deal with additional multiplicative
noise in the Poisson means (sequencing depth). We have tested our
methods across a range of scenarios and compared them with existing
methods. We found that our method performs well on both simulated and
real data, and is not so heavily influenced by outliers as parametric
methods, and successfully adjusts for sequencing depth unlike na\"ive
PCA on the logarithm of the compositional data.

While our method has focused on Poisson noise, so is less
broad than some parametric methods that extend to exponential family
distributions, our work is broader in that it can be applied to
arbitrary transformations of the data, not just the logarithm. The
ideas behind our work should easily be applicable to other families of
noise distribution, particularly discrete. We have also developed
methods for dealing with sequencing depth, that are not available in
other methods. Furthermore, our method is an order of magnitude faster
than parametric methods, because it does not require optimisation of
all parameters.

\subsection{Future Work}

There are a number of avenues for extending this research. Firstly,
for the log transformation, we were not able to find a perfect
unbiased estimator, and instead used an approximation. Further work to
find a better approximation and corresponding estimators should lead
to better performance of the method. Secondly, while we have focussed
on Poisson measurement error correction, the same ideas could be
applicable to other measurement error correction for PCA. Our method
was based on the fact that the Poisson distribution has one parameter,
so the measurement error variance can be directly estimated from the
data. This technique could be applied to other single parameter
measurement error distributions.  Thirdly, Liu {\em et al.} (2018)
used shrinkage methods to improve the finite-sample performance of the
untransformed Poisson PCA estimators; similar ideas could be applied
to improve the finite sample performance of the sequencing depth
corrected estimators and the transformed estimators. Finally, a large
area of research in PCA is the implementation of sparseness
constraints. Since our method provides an unbiassed variance estimator
for the Poisson means, it should be straightforward to implement a
standard regularisation method to ensure consistency in the $p\gg n$
case.

\appendix

\section{Proof of Consistency of Transformed Variance Estimator}\label{ConsistencyProof}

\begin{theorem*}[\ref{Consistency}]
If $e^{2\lambda}f(\lambda)^2=\sum_{n=0}^\infty
\frac{h(n)\lambda^n}{n!}$ is a globally convergent Taylor series,
$\Lambda$ is a non-negative random variable with finite raw moments $\mu_n$
and $\sum_{n=0}^\infty\frac{\mu_n|h_n|}{n!}$ converges, then the
estimator \eqref{TransformedVarianceEstimator} is consistent.
\end{theorem*}

\begin{proof}
Since the estimator \eqref{TransformedVarianceEstimator} is unbiassed,
the strong law of large numbers tells us that this is a consistent
estimator provided ${\mathbb E}\left(g(\mathbf{X})\right)$ and ${\mathbb
  E}\left(\sum_{k=0}^\mathbf{X} (-1)^{\mathbf{X}-k}{\mathbf{X}\choose k}h(k)\right)$ are both
finite.  We have $${\mathbb E}(g(\mathbf{X}))={\mathbb E}({\mathbb
  E}(g(\mathbf{X})|\Lambda))={\mathbb E}(f(\Lambda))$$ which must be finite
since otherwise the variance of $f(\Lambda)$ which we are trying to
estimate will not be finite.

We have that
\begin{align*}
{\mathbb E}\left(\sum_{k=0}^{\mathbf{X}} (-1)^{\mathbf{X}-k}{\mathbf{X}\choose
  k}h(k)\right)&={\mathbb E}\left({\mathbb E}\left(\sum_{k=0}^{\mathbf{X}} (-1)^{\mathbf{X}-k}{\mathbf{X}\choose
  k}h(k)\middle|\Lambda\right)\right)\\
&\leqslant{\mathbb E}\left({\mathbb E}\left(\sum_{k=0}^{\mathbf{X}} {\mathbf{X}\choose
  k}|h(k)|\middle|\Lambda\right)\right)\\
&={\mathbb E}\left(\sum_{k=0}^\infty\frac{\Lambda^k}{k!}|h(k)|\right)
\end{align*}
We are given that 
$\sum_{k=0}^\infty\frac{{\mathbb E}\left(\Lambda^k\right)}{k!}|h(k)|$
converges, so by Fubini's theorem, 
$${\mathbb
  E}\left(\sum_{k=0}^\infty\frac{\Lambda^k}{k!}|h(k)|\right)=\sum_{k=0}^\infty\frac{{\mathbb
    E}\left(\Lambda^k\right)}{k!}|h(k)|$$
is also finite.
\end{proof}

\section{Estimating Conditional Variance}\label{Computingh}

We are interested in the logarithm case, so for large $\mathbf{X}$, we have
$g(\mathbf{X})=\log(\mathbf{X})$. We now want to estimate $\Var(\log(\mathbf{X})|\Lambda)$. We
will estimate it by taking the average of some function $\tilde{h}(X)$ for
each $X$. That is, we want to design a function $\tilde{h}$ such that
$${\mathbb E}(\tilde{h}(X)|\Lambda)\approx \Var(\log(X)|\Lambda)$$
We will assume that $\mathbf{\Lambda}$ is relatively large. We note that $$\Var(\log(\mathbf{X})|\Lambda)=\Var(\log(\mathbf{X})-\log(\Lambda)|\Lambda)=\Var\left(\log\left(1+\frac{\mathbf{X}-\Lambda}{\Lambda}\right)\middle|\Lambda\right)$$
For large $\mathbf{\Lambda}$, we have that $D=\frac{\mathbf{X}-\Lambda}{\Lambda}$ is
almost certain to be much less than 1, so we can use the Taylor
expansion
$$\log(1+D)=D-\frac{D^2}{2}+\frac{D^3}{3}-\cdots$$
We truncate this at the $T$th term, for some $T$, then we have
\begin{align*}
\Var(\log(\mathbf{X})|\Lambda)&\approx {\mathbb
  E}\left(\left(D-\frac{D^2}{2}+\cdots+(-1)^{T-1}\frac{D^T}{T}\right)^2\middle|\Lambda\right)-\left({\mathbb
  E}\left(D\middle|\Lambda\right)-{\mathbb
  E}\left(\frac{D^2}{2}\middle|\Lambda\right)+\cdots+(-1)^{T-1}{\mathbb
  E}\left(\frac{D^T}{T}\middle|\Lambda\right)\right)^2\nonumber\\
&\approx \sum_{n=2}^{2T} (-1)^n{\mathbb E}(D^n|\Lambda)\sum_{i=1}^{(n-1)\land
  T}\frac{1}{i(n-i)}-\sum_{n=2}^{2T} (-1)^n\sum_{i=1}^{(n-1)\land T}\frac{{\mathbb E}(D^i|\Lambda){\mathbb E}(D^{n-i}|\Lambda)}{i(n-i)}
\end{align*}
We can truncate this again at $T$ (i.e. remove all powers of $D$
larger than $T$) to get
\begin{equation}\label{eqhnest}
\Var(\log(\mathbf{X})|\Lambda)\approx \sum_{n=2}^{T} (-1)^n{\mathbb E}(D^n|\Lambda)\sum_{i=1}^{(n-1)}\frac{1}{i(n-i)}-\sum_{n=2}^{T} (-1)^n\sum_{i=1}^{(n-1)}\frac{{\mathbb E}(D^i|\Lambda){\mathbb E}(D^{n-i}|\Lambda)}{i(n-i)}\end{equation}
Now we compute
\begin{align*}
  {\mathbb
    E}(D^n|\Lambda)&=\Lambda^{-n}\sum_{k=0}^n(-1)^{n-k}\binom{n}{k}\Lambda^{n-k}{\mathbb
  E }(\mathbf{X}^k|\Lambda)\\
&=\sum_{k=0}^n(-1)^{n-k}\binom{n}{k}\Lambda^{-k}\sum_{j=0}^k\stirling{k}{j}\Lambda^j\\
&=\sum_{l=0}^n\Lambda^{-l}\sum_{k=l}^{n}(-1)^{n-k}\binom{n}{k}\stirling{k}{k-l}\\
&=\sum_{l=0}^n\Lambda^{-l}\sum_{m=0}^{n-l}(-1)^{m}\binom{n}{m}\stirling{n-m}{n-l-m}\\
\end{align*}
where $\stirling{k}{j}$ is a Stirling number of the second kind. The
second last equation comes from setting $l=k-j$, and the last equation
comes from setting $m=n-k$. We have that
$\binom{n}{m}\stirling{n-m}{n-l-m}$ is the number of partitions of a
set of size $n$ into $n-l$ blocks with $m$ chosen singleton
blocks. This gives us
that $$\sum_{m=0}^{n-l}(-1)^{m}\binom{n}{m}\stirling{n-m}{n-l-m}$$ is
the number of partitions of $n$ with $n-l$ blocks and no singletons,
since for any partition into $n-l$ blocks, of which $M$ are singletons, this
partition is included $\binom{M}{m}$ times in the $m$th term in the
series (once for each subset of $m$ singletons). Therefore the total
term assigned to this partition is
$$\sum_{m=0}^n
(-1)^m\binom{M}{m}=\left\{\begin{array}{ll}1&\textrm{if }M=0\\
0&\textrm{otherwise}\end{array}\right.$$
We will denote the number of partitions with $n-l$ blocks and no
singletons by $\Stirling{n}{n-l}$. We therefore have $${\mathbb E}(D^n|\Lambda)=\sum_{l=0}^n\Stirling{n}{n-l}\Lambda^{-l}=\sum_{k=0}^{\left\lfloor\frac{n}{2}\right\rfloor}\Stirling{n}{k}\Lambda^{k-n}$$
We now substitute this into Equation~\eqref{eqhnest} (noting that for
$n>0$, we have $\Stirling{n}{0}=0$, and for $k>\frac{n}{2}$, we have $\Stirling{n}{k}=0$) to get:
\begin{align*}
\Var(\log(\mathbf{X})|\Lambda)&\approx \sum_{n=2}^{T}(-1)^n\Lambda^{-n}
\left(\sum_{k=1}^{\left\lfloor\frac{n}{2}\right\rfloor}\Stirling{n}{k}\Lambda^k\right)\left(\sum_{i=1}^{(n-1)}\frac{1}{i(n-i)}\right)-\sum_{n=2}^{T}(-1)^n\Lambda^{-n}
\sum_{i=1}^{(n-1)}\frac{\left(\sum_{k=1}^{\left\lfloor\frac{i}{2}\right\rfloor}\Stirling{i}{k}\Lambda^k\right)\left(\sum_{l=1}^{\left\lfloor\frac{n-i}{2}\right\rfloor}\Stirling{n-i}{l}\Lambda^l\right)}{i(n-i)}\\
&= \sum_{k=1}^{T} 
\left(\sum_{n=2k}^{T}(-1)^n\left(\sum_{i=1}^{(n-1)}\frac{1}{i(n-i)}\right)\Stirling{n}{k}\Lambda^{k-n}\right)-
\sum_{m=2}^{\left\lfloor\frac{T}{2}\right\rfloor}\left( \sum_{k=1}^{m-1}\sum_{n=2m}^{T}(-1)^n\Lambda^{m-n}
\sum_{i=2k}^{n+2k-2m}\frac{\left(\Stirling{i}{k}\Stirling{n-i}{m-k}\right)}{i(n-i)}\right)\\
&= \sum_{l=1}^{T} 
\left(\sum_{n=l+1}^{T\land 2l}(-1)^n\left(\sum_{i=1}^{(n-1)}\frac{1}{i(n-i)}\right)\Stirling{n}{n-l}\Lambda^{-l}\right)-
\sum_{l=2}^{T}\Lambda^{-l}\left( \sum_{n=l+2}^{T\land 2l}(-1)^n\sum_{k=1}^{n-l-1}
\sum_{i=2k}^{2l+2k-n}\frac{\left(\Stirling{i}{k}\Stirling{n-i}{n-l-k}\right)}{i(n-i)}\right)
\end{align*}
Recall that $\frac{1}{(\mathbf{X}+1)\cdots(\mathbf{X}+l)}$ is an estimator for $\Lambda^{-l}$, so
substituting in this estimator gives us the estimator $\tilde{h}(\mathbf{X})$ as
$$\sum_{l=1}^{T-1} a_l \prod_{j=1}^{l}(\mathbf{X}+j)^{-1}$$
where
$$a_l=\sum_{n=l+1}^{T\land 2l}(-1)^n\sum_{i=1}^{n-1}\frac{1}{i(n-i)}\left(\Stirling{n}{n-l}-\sum_{k=\left\lceil\frac{i+n}{2}\right\rceil-l}^{\left\lfloor\frac{i}{2}\right\rfloor}\Stirling{i}{k}\Stirling{n-i}{n-l-k}\right)$$
Our Poisson PCA package truncates at $T=9$. The terms $a_l$ can be
precalculated once for all data points, so computation of $\tilde{h}$
is quick.

\section{Alternative Estimator for Conditional Variance}\label{generallowvarianceCVestimator}

Recall that the general estimator for conditional variance of a
transformation $g(X)$ depends on computing an estimator for
$e^{-2\Lambda}\frac{\Lambda^n}{n!}$ for an integer $n$. The unbiased
estimator for this quantity is $t_n=(-1)^{X-n}\binom{X}{n}$.  However,
this estimator can have very high variance, resulting in poor
finite-sample performance. For the special case $g(X)=\log(X)$, we
computed a useable approximation in Appendix~\ref{Computingh}. For the
general transformation case, we do not have a good general method. In
this section, we suggest one possible approach that could be further studied.

The idea is to try to find two observations, $X_1$ and $X_2$ with the same value of
$\mathbf{\Lambda}$. Then we
have $$t'_n(X_1,X_2)=\left\{\begin{array}{ll}1&\textrm{if
}X_1+X_2=n\\0&\textrm{otherwise}\end{array}\right.$$ as an estimator
for $e^{-2\lambda} \frac{(2\lambda)^n}{n!}$. We can then average this
estimator over all pairs of values of $X_1$ and $X_2$ that come from
approximately the same value of $\mathbf{\Lambda}$.

In practice we don't know whether two values $X_1$ and $X_2$ come from
the same value of $\Lambda$. We therefore weight all pairs by how
plausible it is that they have the same value of
$\mathbf{\Lambda}$. We use the weight derived from the likelihood
ratio statistic:
  $$W_{ij}=\frac{(X_i+X_j)^{X_i+X_j}}{2^{X_i+X_j}X_i^{X_i}X_j^{X_j}}$$ and take the
  weighted average of $t'_n(X_1,X_2)$ divided by $2^n$ to get the estimator of $e^{-2\lambda}  \frac{\lambda^n}{n!}$.

  That is, our estimator is
  $$t_n'=\frac{\sum_{\left\{(i,j)\middle|X_i+X_j=n\right\}}
    W_{ij}}{\sum_{i,j=1}^n W_{ij}}$$


\begin{thebibliography}{99}


\bibitem{Arrigo} Arrigo, K. (2005) Marine microorganism and global
  nutrient cycles {\em Nature 437}, 349--355.


\bibitem{Cai} 
Cai, Y., Gu, H., Kenney, T. (2017). Learning Microbial Community Structures with
Supervised and Unsupervised Non-negative Matrix Factorization. {\em
  Microbiome 5}, 110.


\bibitem{caporaso2011moving}
 Caporaso, J. G., Lauber, C. L., Costello, E. K., Berg-Lyons, D.,
  Gonzalez, A., Stombaugh, J., Knights, D, Gajer, P., Ravel, J.,
  Fierer, N., et~al. (2011) 
\newblock Moving pictures of the human microbiome.
\newblock {\em Genome biology}, 12(5):R50.


\bibitem{PLN} 
Chiquet, J., 
Mariadassou, M., Robin, S. (2018)
Variational  Inference  for  Probabilistic
Poisson  PCA {\em ArXiv:1703.06633v5}.
  
\bibitem{CollinsDasguptaSchapire}  Collins, M., Dasgupta, S., Schapire, R.E. (2001). A generalization of
principal components analysis to the exponential family. {\em Advances in Neural Information Processing Systems 14},
 617--624.
  

\bibitem{Fuji} Fujimura, K., Slusher, N., Cabana, M. and Lynch, S. (2010). Role of the gut microbiota in defining human health. {\em Expert review of anti-infective therapy 8}, 435--454.

\bibitem{gorzelak2015methods}
Gorzelak, M.A., Gill, S. K., Tasnim, N.,  Ahmadi-Vand, Z., Jay, M.,
and Gibson, D. L.
 Methods for improving human gut microbiome data by reducing
  variability through sample processing and storage of stool. {\em
    PloS one 10} e0134802, 2015.




\bibitem{Holmes} Holmes, I., Harris, K., Quince, C. (2012).
Dirichlet multinomial mixtures: Generative models for microbial metagenomics.
{\em PLos One 7} 30126.

\bibitem{JolliffeTrendafilovUddin} Jolliffe, I. T.,  Trendafilov,
  N. T. and Uddin M. (2003). A Modified Principal Component Technique Based on the LASSO. { \em Journal of Computational and Graphical Statistics, 12}  531--547.



\bibitem{KnightsD} Knights, D., Kuczynski, J., Charlson, E. S. (2011). Bayesian
community-wide culture-independent microbial source tracking. {\em Nature Methods 8}, 761--763.

\bibitem{Kostic}
Kostic, A. D., Gevers, D., Siljander, H., Vatanen, T., Hy\"{o}tyl\"{a}inen, T.,
H\"{a}m\"{a}l\"{a}inen, A., Peet, A., Tillmann, V., P\"{o}h\"{o}, P., and Mattila, I. (2015).
The Dynamics of the Human Infant Gut Microbiome in Development
and in Progression Toward Type 1 Diabetes {\em Cell Host \& Microbe} 17,
260–273. 


 
\bibitem{LiPaltaShao} Li , L., Palta, M. and Shao, J. (2004).  A
  measurement error model with a Poisson distributed surrogate {\em Statistics in Medicine 23} 2527--2536.

\bibitem{li2015} Li., H. (2015) Microbiome, metagenomics, and high-dimensional compositional data analysis. {\em Annual Review of Statistics and Its Application 2} 73--94.

\bibitem{liudobribanSinger}Liu, L. T., Dobriban, E. and Singer,
  A. (2018) $e$PCA: High dimensional exponential family PCA. {\em
    Annals of Applied Statistics 12} 2121--2150.

\bibitem{Quince}
Quince, C., Lundin, E. E., Andreasson, A. N., Greco, D., Rafter, J., Talley, N.
J., Agreus, L., Andersson, A. F., Engstrand, L., and D’Amato, M. (2013).
The Impact of Crohn’s Disease Genes on Healthy Human Gut Microbiota:
A Pilot Study {\em Gut}, 62, 952–954. 


\bibitem{McMurdie}
McMurdie, P. and
Holmes, S. (2014) Waste not, want not: why rarefying microbiome data
is inadmissible. {\em PLoS Comput. Biol. 10}.


\bibitem{SalmonHarmanyDeledalleWillett} Salmon, J., Harmany, Z., Deledalle, C-A. and Willett, R. (2014). Poisson Noise Reduction with Non-local PCA {\em J. Math. Imaging Vis. 48}, 279--294.

\bibitem{Sek} Sekirov, I., Russell, S., Antunes, L. and Finlay , B. (2010).
Gut microbiota in health and disease. {\em Physiological reviews} 90, 859--904



\bibitem{BioMico} Shafiei, M., Dunn, K. A., Boon, E., MacDonald, S. M.,
  Walsh, D. A., Gu, H., Bielawski, J. P. (2015).
Biomico: a supervised bayesian model for inference of microbial
  community structure {\em Microbiome 3} 1.

\bibitem{BiomeNet}
Shafiei, M., Dunn, K., Chipman, H., Gu, H., Bielawski, J. P. (2014)
Biomenet: A bayesian model for inference of metabolic divergence among
  microbial communities.
{\em PLoS Computational Biology 10}, 1003918.


\bibitem{Turnbaugh}
Turnbaugh, P. J., Hamady, M., Yatsunenko, T., Cantarel, B. L., Duncan, A.,
Ley, R. E., Sogin, M. L., Jones, W. J., Roe, B. A., Affourtit, J. P., Egholm,
M., Henrissat, B., Heath, A. C., Knight, R., and Gordon, J. I. (2009). A
Core Gut Microbiome in Obese and Lean Twins, {\em Nature}, 457, 480–484.


\bibitem{ZouHastieTibshirani} Zou, H., Hastie, T., and Tibshirani, R. (2006) Sparse Principal
Component Analysis. {\em Journal of Computational and Graphical Statistics 15}, 262--286. 

\end{thebibliography}
\end{document}